%% file: bare_jrnl.tex
\documentclass{IEEEoj}
\usepackage{cite}
\usepackage{array}
\usepackage{amsmath,amssymb,amsfonts}
\usepackage{setspace}
\usepackage{algorithmic}
\usepackage{tabularx}
\usepackage{booktabs}
\usepackage{setspace}
\usepackage{stfloats}
\usepackage{multirow}
\usepackage{float}
\usepackage{textcomp}
\usepackage{amsthm}
\usepackage{graphics}
\usepackage{subfigure}
\usepackage[caption=false]{subfig}
\usepackage{setspace}
\usepackage{graphicx,color}

\input{defs}

\def\BibTeX{{\rm B\kern-.05em{\sc i\kern-.025em b}\kern-.08em
    T\kern-.1667em\lower.7ex\hbox{E}\kern-.125emX}}
\AtBeginDocument{\definecolor{ojcolor}{cmyk}{0.93,0.59,0.15,0.02}}
\def\OJlogo{\vspace{-4pt}\hskip-4pt\includegraphics[height=18pt]{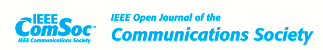}}
\begin{document}
\receiveddate{XX Month, XXXX}
\reviseddate{XX Month, XXXX}
\accepteddate{XX Month, XXXX}
\publisheddate{XX Month, XXXX}
\currentdate{XX Month, XXXX}
\doiinfo{OJCOMS.2024.011100}

\title{A Superdirective Beamforming Approach with Impedance Coupling  and Field Coupling for Compact Antenna Arrays}

\author{{L}iangcheng Han \IEEEmembership{(Graduate Student Member, IEEE)}, Haifan Yin \IEEEmembership{(Senior Member, IEEE)}, Mengying Gao, Jingcheng Xie \IEEEmembership{(Graduate Student Member, IEEE)}}
\affil{School of Electronic Information and Communications, Huazhong University of Science and Technology, Wuhan 430074, China}

\corresp{CORRESPONDING AUTHOR: Haifan Yin (e-mail: yin@hust.edu.cn).}
\authornote{This article was presented in part\cite{DBLP:conf/icc/HanYM22}
at the 56th IEEE International Conference on Communications (IEEE ICC
2022). This work was supported by the National Natural Science Foundation of China under Grant 62071191.}

\begin{abstract}
In most multiple-input multiple-output (MIMO) communication systems, antennas are spaced at least half a wavelength apart to reduce mutual coupling. In this configuration, the maximum array gain is equal to the number of antennas. However, when the antenna spacing is significantly reduced, the array gain of a compact array can become proportional to the square of the number of antennas, greatly exceeding that of traditional MIMO systems. Achieving this "superdirectivity" requires complex calculations of the excitation coefficients (beamforming vector), which is a challenging task. In this paper, we address this problem with a novel double coupling-based superdirective beamforming method. In particular, we categorize the antenna coupling effects to impedance coupling and field coupling. By characterizing  these two coupling in model,  we derive the beamforming vector for superdirective arrays. We prove that the field coupling matrix has the unique solution for an antenna array, and itself has the ability to fully characterize the distorted coupling field. Based on this proven theorem, we propose a method that accurately calculates the coupling matrix using only a number of angle sampling points on the order of the number of antennas. Moreover, a prototype of an independently-controlled superdirective antenna array is developed. Full-wave electromagnetic simulations and real-world experiments validate the effectiveness of our proposed approaches, and superdirectivity is achieved in reality by a compact array with 4 and 8 dipole antennas. 
\end{abstract}

\begin{IEEEkeywords}
superdirectivity, beamforming, coupling matrix, compact array, experimental validation.
\end{IEEEkeywords}
\maketitle
\section{Introduction}
As a key technology of fifth-generation (5G) mobile communication systems, massive MIMO is being deployed globally alongside 5G networks \cite{marzetta2016fundamentals}. According to \cite{marzetta2010noncooperative}, the spectral efficiency significantly increases as the number of base station antennas approaches infinity. However, in practical applications, antenna spacing is typically no less than half a wavelength. This spacing minimizes the coupling effect between antennas, simplifies system design, and restricts the number of antennas accommodated within a given aperture. Recently, to further enhance spectral efficiency, researchers have explored the deployment of ultra-dense antenna arrays at base stations \cite{bjornson2019massive}\cite{pizzo2022fourier} \cite{9110848}, which necessitates consideration of mutual coupling effects. Notably, coupling effects are not entirely detrimental. In compact arrays with $M$ antennas spaced much closer than half a wavelength, strong coupling enables superdirectivity, potentially increasing beamforming gain to $M^2$, compared to the $M$ gain in traditional MIMO theory.

In general, electromagnetic waves emitted by antennas can be divided into propagating waves and evanescent waves \cite{clemmow2013plane}. The coupling between antennas results from the combined effects of these two types of waves. When the antenna spacing is significantly less than half a wavelength, the influence of evanescent waves becomes more pronounced than that of propagating waves, leading to very strong coupling between antennas \cite{hansen2014exact}\cite{altshuler2005monopole}. Traditional communication research often overlooks the role of coupling in an antenna array \cite{8861014}. Theoretically, if the amplitude and phase of each antenna excitation are precisely controlled, this strong coupling can enable superdirectivity \cite{bjornson2019massive}\cite{bloch1953new}. According to conventional communication theory, array gain is proportional to the number of antennas \(M\). However, Uzkov demonstrated that the directivity of an isotropic linear array with \(M\) antennas can reach \(M^2\) when the spacing between antennas approaches zero \cite{uzkov1946approach}. If a base station is equipped with a large number of antennas, the improvement in array gain is even more significant \cite{DBLP:conf/acssc/Marzetta19}.

Despite the remarkable potential of gain improvement introduced by superdirective antenna arrays, there are several challenges that hinder its practical realization. In particular, it appeared that the precision of superdirective beamforming vector can affect the performance of the antenna array\cite{altshuler2005monopole}. And the calculation of the beamforming vector is known as a very challenging problem. This is related to the strong antenna mutual coupling which is hard to characterize in superdirective arrays. The authors of \cite{altshuler2005monopole} initiates the practical measurements of a two-element superdirective array, whose beamforming vector calculation ignores the mutual field coupling between antennas. In \cite{ivrlavc2010toward}, the authors analyzed the performance of communication systems using circuit theory and numerically simulated the superdirectivity of transmit and receive arrays. However, the study did not consider the field coupling between antennas. In addition, it was not validated in practice. In Sec. \ref{IV}, we show that this scheme experiences performance degradation in real systems due to the lack of consideration for field coupling. The authors of \cite{4648403} propose a decoupling array structure to enhance the directivity, which may lead to high hardware complexity, especially when the number of antennas is large. A compact parasitic four-element superdirective array is designed in \cite{clemente2015design}, yet the method depends on the accuracy of simulation tools and may introduce unwanted negative resistance to array design.
 \subsection{Prior Works}
The existing beamforming methods in the MIMO system considering the mutual coupling can be generally divided into two categories. The first category mainly suppresses the mutual coupling\cite{zheng2017robust,selvaraju2018mutual,zhang2019mutual}, while the second category  leverages the mutual coupling to enhance the system performance\cite{williams2020communication,williams2021multiuser,wang2022beamforming,akrout2023super}.

In the first category, the key challenge is how to reduce mutual coupling. Specifically, the authors in \cite{zheng2017robust} presents a robust adaptive beamforming algorithm that addresses the detrimental effects of the mutual coupling. To suppress the mutual coupling through antenna designs, the authors of \cite{selvaraju2018mutual} propose an antenna array with complementary split ring resonator (CSRR) for beamforming applications. 
The authors in \cite{zhang2019mutual} introduces the concept of decoupling ground (DG) to enhance the isolation in massive MIMO antenna arrays. By adjusting the shape of the ground plane under each element, the mutual coupling from free space and the ground plane is canceled out. However, this category of beamforming overlooks the potential benefits of mutual coupling, which may limit further improvements in antenna array performance.

The second category of beamforming mathods view  mutual coupling as a benefit. For instance, the authors of \cite{williams2020communication} derive the expression of coupling-aware beamforming vector, which shows that superdirectivity can be potentially realized in compact antenna arrays. The authors in  \cite{williams2021multiuser} studies the large intelligent surfaces-based multi-user communication scenario considering mutual coupling. The authors point out the practical limitations such as ohmic losses in achieving superdirectivity. The authors of \cite{wang2022beamforming} investigate the beamforming performance of holographic surfaces implemented as arrays of antenna elements with less than half-wavelength spacing packed in finite surface apertures. The numerical results show the superiority of coupling-aware beamforming in comparison to traditional beamforming. In \cite{akrout2023super}, the authors propose a unified model for tightly-coupled massive MIMO antennas leverages mutual coupling to enhance bandwidth and MIMO gains. 

While existing beamforming methods leveraging mutual coupling \cite{williams2020communication,williams2021multiuser,wang2022beamforming,akrout2023super} have shown performance improvements, they have primarily relied on numerical software for validation. However, it is important to note that the coupling induced by near-field evanescent modes can hardly be accurately represented in numerical simulations. In our study, we have found that field coupling significantly impacts the performance of antenna arrays. Thus, to the best of our knowledge, few works have addressed the practical implementation of superdirective beamforming that takes into account both field coupling and impedance coupling.
 \subsection{Our Contributions}
In this paper, we revisit the problem of superdirective beamforming vector calculation by taking a double-coupling effect of antenna arrays into consideration. We find there are two sorts of coupling which we call impedance coupling and field coupling respectively in compact antenna arrays. The impedance coupling represents the  power radiation interaction between antennas, or is regarded as the mutual impedance matrix of the array in free space. While the field coupling leads to the distortion of the radiation pattern of every antenna in the array.
We propose a novel coupling matrix-based method to depict the field coupling of superdirective antenna arrays.  More specifically, our method is based on the fact that the coupling field  can be viewed as the antenna array excited with a specific beamforming vector. By leveraging the electric field containing the coupling information, the field coupling matrix of antenna arrays is derived precisely in simulations. Moreover, in practice the electromagnetic environment is more complex than the ideal simulation conditions. It is meaningful to model the coupling of antenna arrays with realistic measured data.  To this end, we further devise a measurement scheme to compute the coupling matrices of a practical antenna array.

The contributions of this paper are as follows:
\begin{itemize}
    \item  We propose to introduce a field coupling matrix to the superdirective antenna array model. In the traditional superdirective beamforming method, the field coupling is not considered, which impairs the directivity. For an antenna array with small spacing, the field coupling plays a key role in characterizing the interactions between the antennas. Substantial directivity improvements over existing beamforming methods are observed in simulations and measurements.
    \item   We prove that the field coupling matrix has the unique solution for an antenna array, and itself has the ability to fully characterize the distorted coupling field. Based on this proven theorem, we propose a method that accurately calculates the coupling matrix using only a number of angle sampling points on the order of the number of antennas. Simulation results also show that the field coupling matrix computed by our method is precise. 
    \item We further propose a practical method to calculate the coupling matrix based on the realistic measurement data. In practice, there are some conditions that are difficult to be considered by simulations, e.g., the non-ideal manufacturing technology of the antennas, the radiation effect of the coaxial line, the influence of the connectors, etc., which make the actual radiation environment of the antenna array quite different from the ideal simulation.  The proposed method is able to cope with these practical limitations. 
    \item We build a prototype of the superdirective antenna array. To the best of our knowledge, the realization of independently controlled four or  eight-channel superdirective antenna arrays has not been presented so far. A beamforming control board, which offers 7-bit amplitude precision and 8-bit phase precision, is designed to excite the antenna array. The superdirective antenna arrays with four and eight elements are tested in a microwave anechoic chamber. Our proposed methods are validated with this platform in the microwave anechoic chamber, and the gain of directivity over existing methods is confirmed with real-world measurements.

\end{itemize}
In particular, this paper differs from the previous work \cite{DBLP:conf/icc/HanYM22} in several significant ways:

\begin{itemize}
    \item \textbf{Proof of uniqueness for the coupling matrix}: We have added a formal proof of the uniqueness of the coupling matrix, which was not provided in the conference paper. This is a crucial theoretical contribution that enhances the foundation of superdirective beamforming.
    \item \textbf{Extensive simulation experiments}: We have expanded the simulation section with more comprehensive experiments, providing deeper insights and validation of the proposed methods, which were not included in the previous version.
    \item \textbf{Reduction of angular sampling points}: In this paper, we introduce a novel approach that reduces the number of angular sampling points needed to compute the coupling matrix \( \mathbf{C} \), improving the efficiency of calculations while maintaining comparable accuracy. This was not explored in the conference paper.
    \item \textbf{Practical implementation}: We propose a method for practically implementing superdirective beamforming, providing a clear path for real-world application, which was absent in the conference paper.
    \item \textbf{Superdirective antenna testing system}: A significant addition is the design and construction of a superdirective antenna testing system, which validates the performance and applicability of the proposed methods in real environments.
\end{itemize}
\subsection{Organization and Notation}
This paper is organized as follows: In Sec. \ref{II} we introduce the superdirective beamforming method based on array theory. In Sec. \ref{III} we propose the coupling matrix and full-wave simulation-based superdirective beamforming approach. In Sec. \ref{IV} the numerical results including simulations and practical measurements are shown. Finally, conclusions are drawn in Sec. \ref{V}.

Notations: The boldface font stands for vector and matrix. ${{\bf{X}}^\dag },{{\bf{X}}^T},{{\bf{X}}^*},$ and ${{\bf{X}}^H}$ denote the Moore–Penrose pseudoinverse, transpose, conjugation, and Hermitian transpose of $\bf{X}$, respectively. ${\mathbb{C}^{a \times b}}$ is a matrix space with $a$ rows and $b$ columns. $\left| {\bf{x}} \right|$ denotes the absolute value of $\bf{x}$ and ${\left\| {\bf{X}} \right\|_2}$ is the second-order induced norm of $\bf{X}$. $ \buildrel \Delta \over = $ refers to the definition symbol. 

\section{Beamforming of superdirective arrays based on array theory}\label{II}
According to \cite{uzkov1946approach}, only when at least three requirements are met can an antenna array achieve superdirectivity. The first is that the spacing  between two neighboring antennas should be less than half a wavelength, which is simple to accomplish.
Second, superdirectivity requires the precise calculation of beamforming vectors, which in turn relies on accurately modeling the coupling between antennas. This coupling is influenced by various factors, such as the relative positions of the antennas, the coaxial cables, the antenna mounting structures, etc. The method proposed in this paper addresses this challenge by taking the field coupling into account and measuring the electromagnetic environment information directly. By obtaining accurate data from the environment, the beamforming vectors can be calculated with greater precision, ensuring that the conditions for superdirectivity are met. Third, implementing superdirectivity in real-world scenarios necessitates extremely precise control over the amplitude and phase of signals across multiple antenna elements. This precision is difficult to achieve due to limitations in hardware, such as phase noise, thermal drift, and manufacturing tolerances, which can lead to discrepancies between the intended and actual beamforming vectors.  One potential solution to this challenge is the use of high-precision RF amplitude and phase control chips. These chips can provide the necessary accuracy by offering fine-grained control over the signal parameters. Additionally, advanced calibration techniques and real-time adaptive algorithms can be employed to continuously adjust the beamforming vectors, compensating for any deviations caused by hardware imperfections or environmental changes.
This section will present the traditional method to calculate the  superdirective beamforming vector based on array theory.

Consider a uniform linear array consisting of $M$ antennas, each having an identical pattern function $g(\theta, \phi)$, where $\theta$ and $\phi$ represent the far-field coordinates in a spherical coordinate system. For simplicity, assume the first antenna is positioned at the origin of the Cartesian coordinate system, and the remaining antennas are uniformly distributed along the positive half of the $y$-axis with a spacing of $d$. As a result, the complex far-field pattern of this antenna array is given by
\begin{align}\label{array_pattern_function}
f(\theta,\phi)=\sum_{m=1}^Ma_mg(\theta,\phi)e^{jk\mathbf{\hat r}\cdot \mathbf{r}_m},
\end{align}
where $a_m$ represents the complex excitation coefficient corresponding to the current on the $m$-th antenna, $k = 2\pi/\lambda$ is the wave number, $\mathbf{\hat{r}}$ is the unit vector in the far-field direction $(\theta, \phi)$ in the spherical coordinate system, and $\mathbf{r}_m$ denotes the position of the $m$-th antenna. Consequently, the directivity factor at $(\theta_0, \phi_0)$ can be calculated as
\begin{align}\label{theta0_phi0_pattern_function}
  D(\theta_0,\phi_0)= \frac{|\sum_{m=1}^{M}a_mg(\theta_0,\phi_0)e^{jk\mathbf{\hat r}_0\cdot\mathbf{ r}_m}|^2}{\frac{1}{4\pi}\int_S |\sum_{m=1}^{M}a_mg(\theta,\phi)e^{jk\mathbf{\hat r}\cdot \mathbf{r}_m}|^2dS},
\end{align}
where $\mathbf{\hat{r}}_0$ is the unit vector in the direction $(\theta_0,\phi_0)$. 

To derive the maximum directivity of an antenna array, the expression in \eqref{theta0_phi0_pattern_function} is complex and requires simplification. The denominator of \eqref{theta0_phi0_pattern_function} can be expanded as
\begin{align}\label{denominator}
&\frac{1}{4 \pi} \int_{0}^{2 \pi} \int_{0}^{\pi}\left|\sum_{m=1}^{M} a_{m} g(\theta, \phi) e^{j k \mathbf{\hat{r}} \cdot \mathbf{r}_m}\right|^{2} \sin \theta d \theta d \phi \notag\\
&=\frac{1}{4 \pi} \int_S \sum_{n=1}^M \sum_{m=1}^{M} a_{n} a_{m}^{*}|g(\theta,\phi)|^2 e^{j k \mathbf{\hat{r}} \cdot \mathbf{r}_n} e^{-j k \mathbf{\hat{r}} \cdot \mathbf{r}_m} dS \notag\\
&= \sum_{n=1}^{M} \sum_{m=1}^{M} a_{n} a_{m}^{*} \frac{1}{4 \pi}\int_S |g(\theta,\phi)|^2e^{j k \mathbf{\hat{r}} \cdot \mathbf{r}_n} e^{-j k \mathbf{\hat{r}} \cdot \mathbf{r}_m} dS,
\end{align}
where the $a_n^*$ denotes the conjugate of the complex value $a_n$. 

For the integral terms in \eqref{denominator}, we present the following equation:
\begin{align}\label{zmn}
z_{mn} \buildrel \Delta \over = \frac{1}{4\pi}\int_0^{2\pi}\int_0^\pi |g(\theta,\phi)|^2e^{jk\mathbf{\hat r}\cdot\mathbf{r}_m }e^{-jk\mathbf{\hat r}\cdot \mathbf{r}_n}\sin\theta d\theta d\phi.
\end{align}
Since the power radiated by the antenna array is active power, $z_{mn}$ represents the real part of the normalized mutual impedance between the $m$-th and $n$-th antennas \cite{zucker_antenna_1969}. Thus, \eqref{denominator} can be rewritten as
     \begin{align}\label{simplify_denominator}
  \setlength{\abovedisplayskip}{4pt}
 &\sum_{m=1}^{M}\sum_{n=1}^{M}a_ma^*_n\frac{1}{4\pi} \int_S |g(\theta,\phi)|^2e^{jk\mathbf{\hat r}\cdot \mathbf{r}_m }e^{-jk\mathbf{\hat r}\cdot \mathbf{r}_n}dS \notag  \\
  &=\sum_{m=1}^{M}\sum_{n=1}^{M}a_ma^*_nz_{mn}.
\end{align} 

For the simplicity of notations, \eqref{theta0_phi0_pattern_function} is further rewritten using two vectors $\mathbf{a},\mathbf{e}\in \mathbb{C}^{M\times 1}$
\begin{align}\label{vector_D}
D=\dfrac{\mathbf{a}^T\mathbf{e}\mathbf{e}^H\mathbf{a}^*}{\mathbf{a}^T\mathbf{Z}\mathbf{a}^*},
\end{align}
where $\mathbf{a}$ denotes the beamforming vector
\begin{align}\label{vector_a}
	\mathbf{a}=\left[ a_1,\,a_2,\,\cdots ,\,a_M \right]^T,
\end{align}
and
\begin{align}\label{vector_e}
	\mathbf{e}=\left[ e^{jg\mathbf{\hat{r}}\cdot \mathbf{r}_1}g(\theta,\phi),\,e^{jk\mathbf{\hat{r}}\cdot \mathbf{r}_2}g(\theta,\phi),\,\cdots ,\,e^{jk\mathbf{\hat{r}}\cdot \mathbf{r}_M}g(\theta,\phi) \right]^T.
\end{align}
Note that secondary radiation caused by other objects around the antenna is not considered in this paper. $\mathbf{Z}\in \mathbb{C}^{M\times M}$ stands for the real part of the normalized impedance matrix, which is symmetric. It can be considered as the impedance coupling between antennas
\begin{align}\label{impedance_matrix}
\mathbf{Z}=\left[\begin{array}{ccc}
z_{11} & \ldots & z_{1 M} \\
\vdots & \ddots & \vdots \\
z_{M 1} & \cdots & z_{M M}
\end{array}\right].
\end{align}
Note that \eqref{vector_D} is in the form of a Rayleigh quotient. The problem of maximizing directivity can be addressed by taking the derivative of \eqref{vector_D} with respect to the vector $\mathbf{a}$ as follows:
\begin{equation}\label{derivative}
  \frac{\partial D}{\partial \mathbf{a}}
=\frac{2\mathbf{e}\mathbf{e}^H\mathbf{a}^*}{\mathbf{a}^T\mathbf{Z}\mathbf{a}^*}-2D\frac{\mathbf{Z}\mathbf{a}^*}{\mathbf{a}^T\mathbf{Z}\mathbf{a}^*}.
\end{equation}
Letting the above formula equal to zero yields
\begin{align}\label{forcezero}
  \mathbf{e}\mathbf{e}^H\mathbf{a}^*=D\mathbf{Z}\mathbf{a}^*.
\end{align}
It can be observed that \eqref{forcezero} represents a generalized eigenvalue problem in the form $\mathbf{A}\mathbf{x} = \lambda \mathbf{B}\mathbf{x}$, where $\mathbf{A}$ and $\mathbf{B}$ are matrices, $\mathbf{x}$ is the generalized eigenvector of $\mathbf{A}$ and $\mathbf{B}$, and $\lambda$ is the corresponding generalized eigenvalue. By multiplying both sides of \eqref{forcezero} by $\mathbf{Z}^{-1}$, we obtain
\begin{align}\label{general}
  \mathbf{Z}^{-1}\mathbf{e}\mathbf{e}^H\mathbf{a}^*=D\mathbf{a}^*.
\end{align}
The eigenvalue of the above equation has a unique solution, as shown in \cite{zucker_antenna_1969}. Given that \(\rm{rank}(\mathbf{e}\mathbf{e}^H) = 1\), the following relationship holds:
\begin{align}
    \rm{rank}(\mathbf{Z}^{-1}\mathbf{e}\mathbf{e}^H)\leq 1.
\end{align}
In addition, if $\rm{rank}(\mathbf{Z}^{-1}\mathbf{e}\mathbf{e}^H)=0$, it implies that the maximum value of \eqref{vector_D} is 0, which is not consistent with reality, therefore
\begin{align}
    \rm{rank}(\mathbf{Z}^{-1}\mathbf{e}\mathbf{e}^H)= 1.
\end{align}
Hence, the only one non-zero eigenvalue of $\mathbf{Z}^{-1}\mathbf{e}\mathbf{e}^H$ is the maximum value of the directivity factor $D_{\max} $. Hence,  \eqref{general} can be rewritten as
\begin{align}\label{generalmax}
  \mathbf{Z}^{-1}\mathbf{e}\mathbf{e}^H\mathbf{a}^*=D_{\max}\mathbf{a}^*.
\end{align}
Since
\begin{align}\label{hh}
  \mathbf{Z}^{-1}\mathbf{e}\mathbf{e}^H\mathbf{a}^*&=\mathbf{Z}^{-1}\mathbf{e}(\mathbf{e}^H\mathbf{a}^*) \notag \\
  &=\xi\mathbf{Z}^{-1}\mathbf{e},
\end{align}
where $\xi = \mathbf{e}^H \mathbf{a}^*$ is a scalar. Therefore, the beamforming vector that maximizes the directivity factor of the antenna array can be expressed as
\begin{align}\label{aa}
  \mathbf{a}=\frac{\xi}{D_{\max}}\mathbf{Z}^{-1}\mathbf{e}^* = \mu \mathbf{Z}^{-1}\mathbf{e}^*,
\end{align}
where the scalar $\mu$ is defined as $\mu = \frac{\xi}{D_{\max}}$. By substituting \eqref{aa} into \eqref{vector_D}, we obtain the maximum directivity factor as
\begin{align}\label{maxD}
  D_{\max}=\mathbf{e}^H\mathbf{Z}^{-1}\mathbf{e}.
\end{align}
However, the above derivation process neglects the field coupling between the antennas. In practice, strong field coupling occurs when the antenna spacing is small, causing the radiation pattern of one antenna to be distorted by the influence of surrounding antennas. In \eqref{array_pattern_function}, if $a_n$ is set to 1 and $a_m$ for $m=2,\cdots,M, m \ne n$, is set to 0, then $f(\theta,\phi)$ would be $a_n g(\theta,\phi)e^{jk\mathbf{\hat{r}} \cdot \mathbf{r}_n}$, indicating that the radiation pattern of the $n$-th antenna is unaffected by other antennas, which is unrealistic. Given the small spacing required by superdirective antenna arrays, mutual field coupling cannot be ignored.
As a result, the beamforming vector based on the traditional approach \eqref{aa} may not be suitable for achieving maximum directivity in an antenna array. To address this and achieve more realistic superdirectivity, we propose a superdirective array analysis method that incorporates both impedance coupling and field coupling matrices. 
\section{Proposed superdirective beamforming based on double couplings}\label{III}
    In this section, we  propose the superdirective beamforming approach which considers both the impedance coupling effect and the field coupling effect. We first show how to obtain the superdirective beamforming vector with the full-wave simulations, which is the foundation of our superdirective beamforming realization method in practice. Then, since there are still some factors not considered in the simulations, we further propose a  measurement-based approach to acquire the impedance coupling and field coupling matrices, which will help produce superdirectivity in practice. 

\subsection{Full-wave simulation-based acquisition of the superdirective beamforming vector}

Since the superdirectivity is produced only when the antenna spacing is small, the field coupling effect should not be ignored when deriving  the beamforming vector.
In this section, we propose to characterize the field coupling effect between antennas with a coupling matrix, which is obtained using  full-wave simulation tools.

Taking into account the mutual field coupling between antennas, a new pattern function \(l(\theta, \phi)\) for the superdirective array, based on field coupling coefficients, can be expressed as follows \cite{1330259}:
\begin{align}\label{couplefxt}
l(\theta,\phi)=\sum_{m=1}^{M}\sum_{n=1}^{M}c_{nm}a_mg(\theta,\phi)e^{jk\mathbf{\hat{r}}\cdot\mathbf{r}_n},
\end{align}
where $c_{nm}$ represents the field coupling coefficient between the $m$-th and $n$-th antennas. Incorporating these coupling coefficients into the antenna array model allows for the quantification of the interactions between antennas.

For ease of analysis, the field coupling coefficients are represented in matrix form as
\begin{equation}
   \begin{aligned}\label{coupling_matrix}
\mathbf{C}=\left[ \begin{matrix}
	c_{11}&		...&		c_{1M}\\
	\vdots&		\ddots&		\vdots\\
	c_{M1}&		...&		c_{MM}\\
\end{matrix} \right].
\end{aligned}
\end{equation}

The electric field $\mathbf{E}(\theta ,\phi )$ radiated by the antenna in the far-field region is denoted by
\begin{align}
  \mathbf{E}(\theta ,\phi 
)=[E_{\hat{\theta}}(\theta,\phi),E_{\hat{\phi}}(\theta,\phi)],
\end{align}
with the subscripts $\hat{\theta}$ and $\hat{\phi}$ denoting the $\theta$-component and $\phi$-component  respectively. The electric field in space is sampled and vectorized to obtain the vector
\begin{align}\label{barE}
\overline{\boldsymbol{\varepsilon }}= [&E_{\hat\theta}\left( \theta _1,\phi _1 \right) ,E_{\hat\phi}\left( \theta _1,\phi _1 \right) ,E_{\hat\theta}\left( \theta _2,\phi _2 \right) ,E_{\hat\phi}\left( \theta _2,\phi _2 \right), \notag\\&\cdots,
E_{\hat\theta}\left( \theta _P,\phi_P\right) ,E_{\hat\phi}\left( \theta _P,\phi _P \right)  ] ^T,
\end{align}
where $P$ is the number of angular sampling points, $\overline{\boldsymbol{\varepsilon }}$ is thus a column vector of size $2P\times 1$.

The detailed steps for calculating the field coupling matrix are outlined below. 

The first step is to obtain the electric field radiated by the antenna array without considering the field coupling effect between antennas, assuming the field coupling matrix \(\mathbf{C}\) is an identity matrix, and only one antenna of the array is excited. Specifically, a single antenna, modeled in full-wave simulation software, is initially placed at the origin of the coordinate system. The electric field \(\mathbf{e}_{s1} \in \mathbb{C}^{2P \times 1}\), as shown in \eqref{barE}, is then obtained through simulation. Next, the antenna is moved to \([0,0,d]\), and the electric field \(\mathbf{e}_{s2}\) is similarly obtained. This procedure continues, moving the antenna a distance \(d\) each time along the positive half-axis of the \(z\)-axis to measure its electric field at various positions. This process is repeated \(M\) times until the antenna is finally placed at \([0,0,(M-1)d]\). Consequently, the set of electric fields for a single antenna at different positions is given by
\vspace{-0.1cm}
\begin{align}\label{Es}
 	\mathbf{E}_s=\left[\mathbf{e}_{s1},\,\mathbf{e}_{s2},\,\mathbf{e}_{s3},\,\cdots,\,\mathbf{e}_{sM}\right],
 \end{align}
 where $\mathbf{E}_s$ is a matrix of size $2P\times M$.

The second step involves obtaining the electric field radiated by the antenna array while accounting for the field coupling effect between antennas. Again, only one antenna of the array is excited, while the others are terminated with matched loads. Specifically, a uniform linear array consisting of \(M\) antennas with spacing \(d\) is modeled in full-wave simulation software, where the \(m\)-th antenna is positioned at \([0,0,(m-1)d]\). Then, only the \(m\)-th antenna (\(m = 1, \cdots, M\)) is excited to radiate the electric field \(\mathbf{e}_{cm}\). In full-wave simulation, \(\mathbf{e}_{cm}\) includes the mutual field coupling effect between antennas, representing the superposition of the electric fields radiated by all antennas. Define \(\mathbf{E_c} \in \mathbb{C}^{2P \times M}\) as
\begin{align}\label{Q}
\mathbf{E}_c=[\mathbf{e}_{c1},\,\mathbf{e}_{c2},\,\cdots,\,\mathbf{e}_{cM}],
\end{align}
where \(\mathbf{E}_c\) represents the set of electric fields radiated by a single antenna in a realistic antenna array, taking into account the inter-element field coupling effect. Consequently, the field coupling matrix \(\mathbf{C}\) serves as a bridge converting \(\mathbf{E}_s\) to \(\mathbf{E}_c\), and the field coupling coefficients can be determined by solving the following linear equations:
\begin{align}
    \left\{ \begin{array}{c}
	{\mathbf{e}_{{c}}}_1=c_{11}{\mathbf{e}_{{s}}}_1+c_{21}{\mathbf{e}_{{s}}}_2+\cdots +c_{M1}{\mathbf{e}_{{s}}}_M\\
	{\mathbf{e}_{{c}}}_2=c_{12}{\mathbf{e}_{{s}}}_1+c_{22}{\mathbf{e}_{{s}}}_2+\cdots +c_{M2}{\mathbf{e}_{{s}}}_M\\
	\vdots\\
	{\mathbf{e}_{{c}}}_M=c_{1M}{\mathbf{e}_{{s}}}_1+c_{2M}{\mathbf{e}_{{s}}}_2+\cdots +c_{MM}{\mathbf{e}_{{s}}}_M\\
\end{array} \right. .
\end{align}
\addtolength{\topmargin}{0.01in}
Hence, the field coupling matrix can be calculated as
\begin{align}\label{Cmatrix}
	\mathbf{C}=\mathbf{E}_s^{\dagger}\mathbf{E}_c.
\end{align}
Note that \(\mathbf{C}\) is not a symmetric matrix, meaning \(c_{mn}\) is not necessarily equal to \(c_{nm}\). This asymmetry arises because the field coupling effect from the \(m\)-th antenna to the \(n\)-th antenna and vice versa is influenced by the number of nearby antennas at their respective positions.

\begin{theorem}\label{thm1}

 The field coupling matrix $\mathbf{C}$ of a dipole antenna array has a unique solution, and itself can fully characterize radiation pattern distortion.
\end{theorem} 
\begin{proof}
The proof can be found in Appendix \ref{Appen. A}.
\end{proof}
\begin{figure}[htbp]
  \centering
  \includegraphics[width=3.4in]{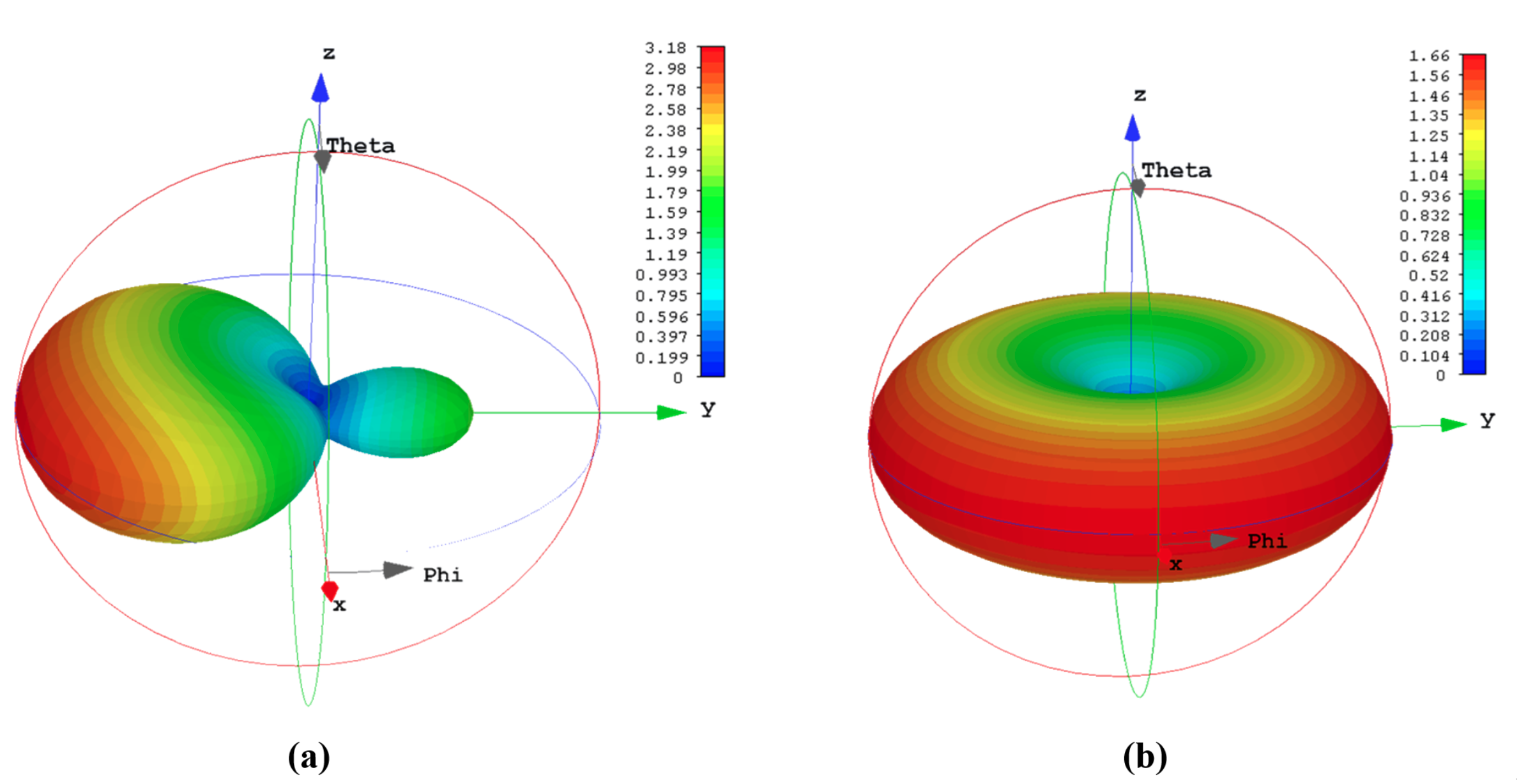}  
  \caption{The antenna pattern (a) with distortion due to the field coupling, (b) without distortion. }\label{pattern}
\end{figure}
When only one antenna is excited in the array, the distorted pattern due to the field coupling is shown in Fig. \ref{pattern} (a). The array configuration is $M=4$ and $d=0.1\lambda$. More detailed parameters of the antenna will be shown in Sec. \ref{IV}. It can be found the pattern is irregular and far different from the antenna pattern without distortion as shown in Fig. \ref{pattern} (b). According to Theorem \ref{thm1}, the field coupling matrix can be used to fully characterize the distortion of the radiation pattern, e.g., the difference between  Fig. \ref{pattern} (a) and Fig. \ref{pattern} (b). The intuition is that the distorted electric field can be uniquely represented by the electric field  $\mathbf{e}_{c_n}$ \cite{hald1988spherical}, which is also uniquely determined by the matrix $\mathbf{C}$ and the undistorted field ${\mathbf{E}}_s$ illustrated by Fig. \ref{pattern} (b). Due to the symmetry in a linear antenna array and Theorem \ref{thm1}, the field coupling matrix $\mathbf{C}$ can be uniquely determined with a reduced number of angular sampling points, as the following corollary shows.
\begin{corollary}\label{coro1}

 For an even number of antennas $M$, the  number of angular sampling points required to determine $\mathbf{C}$ is
    \begin{equation}
        P \geq \frac{M}{2}.
    \end{equation}
    
  For an odd number of antennas $M$, the  number of angular sampling points required to determine $\mathbf{C}$ is:
\begin{equation}
    P \geq M
\end{equation}

\end{corollary}
\begin{proof}
    For both even and odd \( M \), the columns of \(\mathbf{C}\) exhibit symmetry: the \( i \)-th column is the reverse of the \((M - i + 1)\)-th column. This means \( c_{ji} = c_{(M-j+1)(M-i+1)} \) for columns that are not the middle column.
    For an array with \( M \) elements, the total number of elements in the coupling matrix \(\mathbf{C}\) is \( M^2 \).
    Due to the reverse symmetry property of \(\mathbf{C}\), for each column \( i \) (where \( i \neq \frac{M + 1}{2} \) if \( M \) is odd), the entries are related to those in the \((M - i + 1)\)-th column. Therefore, these columns do not need to be solved independently; solving one provides the values for its symmetric counterpart. At a single measurement angle, we can excite each antenna sequentially to obtain \( M \) equations (one for each antenna excitation). For an even number of antennas $M$, since there are \( \frac{M^2}{2} \) independent variables and each angle provides \( M \) equations, the minimal number of angles \( P_{\text{even}} \) required is:
     \[
     P_{\text{even}} = \frac{\frac{M^2}{2}}{M} = \frac{M}{2}.
     \]
      For an odd number of antennas $M$, since the middle column does not exhibit symmetry, a full set of \( M \) angles is required for the calculation.
\end{proof}

By leveraging the property demonstrated in the corollary above, we can update the coupling matrix $\mathbf{C}$ in dynamic and complex environments using measurements from only a few angles. This approach offers several benefits, including reduced measurement time, lower computational costs, and the ability to adapt in  real-time. Furthermore, it increases the practicality of deploying antenna arrays in environments where full-field data collection is impractical or costly.

According to \eqref{vector_D} and \eqref{couplefxt}, the directivity factor based on the field coupling matrix is thus
\begin{align}\label{coupleD}
	D_{c}=\dfrac{(\mathbf{C}\mathbf{b})^T\mathbf{e}\mathbf{e}^H(\mathbf{C}\mathbf{b})^*}{(\mathbf{C}\mathbf{b})^T\mathbf{Z}(\mathbf{C}\mathbf{b})^*},
\end{align}
where $\mathbf{b}$ is the vector of excitation coefficients that maximizes the directivity $D_c$ based on our proposed method. It is observed that this equation is also in the form of Rayleigh quotient. Thus it is analogous to \eqref{aa} to obtain the beamforming vector  $\mathbf{b}$ as
\begin{align}\label{mutuala}
	\mathbf{b}=\zeta\mathbf{C}^{-1}\mathbf{Z}^{-1}\mathbf{e}^*,
\end{align}
where $\zeta$ is a constant, which is determined by the total power constraint. 

When field coupling is ignored, the theoretical beamforming coefficients $\mathbf{a}$ are directly applied to the antenna array. However, due to the presence of field coupling, the actual currents induced on each antenna element are influenced by the currents on neighboring elements. This interaction causes the actual excitation coefficients to deviate from the theoretical ones, resulting in the actual excitation being represented as $\mathbf{Ca}$, where $\mathbf{C}$ is the coupling matrix.

\textbf{Impact on directivity}:
     The directivity of the array is directly influenced by the beamforming coefficients. The theoretical directivity $D_{\text{theory}}$ is calculated based on the assumption that the excitation coefficients are precisely $\mathbf{a}$. However, due to coupling, the actual directivity $D_{\text{actual}}$ is derived from the modified coefficients $\mathbf{Ca}$. The difference in directivity can be expressed as:
     \begin{align}
         \Delta D &= D_{\text{theory}} - D_{\text{actual}}\notag \\
         &= \dfrac{\mathbf{a}^T\mathbf{e}\mathbf{e}^H\mathbf{a}^*}{\mathbf{a}^T\mathbf{Z}\mathbf{a}^*}-\dfrac{(\mathbf{C}\mathbf{a})^T\mathbf{e}\mathbf{e}^H(\mathbf{C}\mathbf{a})^*}{(\mathbf{C}\mathbf{a})^T\mathbf{Z}(\mathbf{C}\mathbf{a})^*}
     \end{align}
      This equation illustrates how the directivity can degrade due to the influence of the coupling matrix $\mathbf{C}$.

\textbf{Impact on beamforming accuracy}: Beamforming accuracy depends on the alignment of the main lobe with the desired direction and the suppression of sidelobes. The misalignment caused by coupling effects leads to deviations in the main lobe direction and increases in sidelobe levels. The deviation in the beam pattern can be quantified by evaluating the difference in the resulting beam pattern power distribution, which can be represented as:
   \begin{align}
       &\Delta F \notag \\
       &= \scalebox{0.9} 10\log_{ 10}\left(\frac{1}{L}\sum_{\theta}|F_{\text{theory}}(\theta) - F_{\text{actual}}(\theta)|^2\right) \notag \\
       &= \scalebox{0.9}{$10\log_{ 10}\left(\frac{1}{L}\sum_{\theta}|\dfrac{\mathbf{a}^T\mathbf{e}(\theta)\mathbf{e}(\theta)^H\mathbf{a}^*}{\mathbf{a}^T\mathbf{Z}\mathbf{a}^*}-\dfrac{(\mathbf{C}\mathbf{a})^T\mathbf{e}(\theta)\mathbf{e}(\theta)^H(\mathbf{C}\mathbf{a})^*}{(\mathbf{C}\mathbf{a})^T\mathbf{Z}(\mathbf{C}\mathbf{a})^*}|^2\right)$},
   \end{align}
      where $L$ is the number of sample points. This shows that the sidelobe levels and the direction of the main lobe can be significantly altered due to field coupling, thus impacting the overall beamforming performance.

To analyze the effect of ohmic loss on the antenna array, it is assumed that the efficiency of each antenna is $\eta_a$. The normalized loss resistance induced antenna efficiency is\cite{zucker_antenna_1969}
\begin{align}\label{rloss}
    r_\text{loss} = \frac{1-\eta_a}{\eta_a}.
\end{align}
Thus, the gain of the antenna array can be written as
\begin{align}\label{gain}
    G=\frac{(\mathbf{C b})^{T} \mathbf{e e}^{H}(\mathbf{C b})^{*}}{(\mathbf{C b})^{T} (\mathbf{Z}+r_{\text {loss }} \mathbf{I}_{M})(\mathbf{C b})^{*}}
\end{align}
where $\mathbf{I}_{M}$ is an identity matrix of size $M\times M$.

In practice, the electromagnetic properties of the antenna array always differ from the idealized simulations. Some practical factors, such as the manufacturing of the antenna array, the influence of coaxial cable radiation, etc. will affect the radiation pattern and the directivity of the antenna array. 
To obtain the realistic impedance and field coupling of antenna arrays, a measurement-based coupling matrix calculation method is proposed as follows.

\subsection{Measurement-based computation of the superdirective beamforming vector}
The essence of our proposed superdirective beamforming method is the two coupling matrices, i.e., the field coupling matrix and the impedance matrix. However, the full-space electric field data needed to calculate the two coupling matrices cannot be accessed easily like in the simulation tools. Moreover, some realistic factors may not be easily modeled in simulations. As a result, we propose a practical method to obtain the coupling matrices. This method takes the realistic factors into consideration, and the input data needed is also accessible. In the following, we will show the details of this method in the context of a compact dipole antenna array. Nevertheless, this method can be applied to other types of antennas as well. 
Without loss of generality, the principal radiation plane of the dipole is the H-plane ($\theta=90^o$) in our setting. We only use the electric field data of the H-plane to calculate the coupling matrices. 
Such a simplification facilitates the measurements, which can be done easily in an anechoic chamber with a rotating platform.  Even with only part of the electric field data, the effectiveness of this method has been proven in simulation and practice. 

First consider a z-directed dipole antenna located at the origin with current field $\mathbf{J}$, the Green's function to deduce the electric field radiated by this antenna is written as \cite{harrington_time-harmonic_2001}
\begin{align}\label{GF}
\mathbf{G}(\mathbf{r})=-j \frac{\eta}{2 \lambda}\left(\mathbf{I}_3+\frac{1}{k^2} \nabla \nabla^T\right) \frac{\mathrm{e}^{-i k \|\mathrm{r}\|_2 }}{\|\mathrm{r}\|_2},
\end{align}
where $\mathbf{I}_3$ is an identity matrix of size $3\times 3$, $\nabla=\left[\frac{\partial}{\partial x} , \frac{\partial}{\partial y}, \frac{\partial}{\partial z}\right]^T$ is the gradient operator, $\mathbf{r}$ is the coordinate of the far-field region. The electric field can thus be expressed as
\begin{align}\label{Er}
    \mathbf{E}(\mathbf{r})=\int_\Omega \mathbf{G}\left(\mathbf{r}-\mathbf{r}^{\prime}\right) \mathbf{J}\left(\mathbf{r}^{\prime}\right) \mathrm{d} \mathbf{r}^{\prime},
\end{align}
where $\Omega$ is the volume of the antenna. It can be found from \eqref{GF} and \eqref{Er} that  when the antenna is moved to $[0,d,0]$, the variation of the radiated electric field is caused by the term $\frac{\mathrm{e}^{-i k \|\mathrm{r}\|_2 }}{\|\mathrm{r}\|_2}$ of \eqref{GF}. The Green's function is changed to
\begin{align}\mathbf{G}(\mathbf{\bar{r}})=-j \frac{\eta}{2 \lambda}\left(\mathbf{I}_3+\frac{1}{k^2} \nabla \nabla^T\right) \frac{\mathrm{e}^{-i k \|\mathrm{r-d\sin\theta\sin\phi}\|_2 }}{\|\mathrm{r-d\sin\theta\sin\phi}\|_2}.\end{align}

It is obvious that both the amplitude and phase of the electric field have been changed. However, since the far-field is the considered application scenario,  the antenna spacing satisfies $d\ll \|\mathrm{r}\|_2$, which indicates that $\|\mathrm{r-d\sin\theta\sin\phi}\|_2\approx \|\mathrm{r}\|_2$. Hence, the change of far-field radiation amplitude due to the variety of antenna positions in the array is negligible, which simplifies our measurement scheme.

Consider a uniform linear antenna array with $M$ antennas, the position of each antenna is sequentially numbered with \#1, \#2, $\cdots$, \#$M$. The impedance coupling matrix  of the antenna array, i.e., the $\mathbf{Z}$ matrix in \eqref{impedance_matrix}, is first calculated. The specific measurement method is to first place a dipole antenna at the center of the rotating platform to measure its far-field amplitude pattern ${\Lambda_s}(\phi)$, where $\phi$ varies from $-180^\circ$ to $180^\circ$. Next, this antenna is placed at \#1 and its phase pattern $\Psi_{s_1}(\phi)$ is measured. Similarly, this antenna is then placed at \#2, \#3, $\cdots$, \#$M$, and the corresponding phase patterns $\Psi_{s_2}(\phi)$, $\Psi_{s_3}(\phi)$, $\cdots$, $\Psi_{s_M}(\phi)$ are measured. According to \eqref{zmn} and the fact that $\theta=90^o$, the $i$-th row and $j$-th column of the impedance coupling matrix $\mathbf{Z}$ can be calculated as
\begin{equation}
  \begin{aligned}\label{zij}
    {{z}_{ij}}=\sum\limits_{\phi }{{{(\Lambda_s(\phi ))}}}{{e}^{j\cdot \Psi_{s_i}(\phi )}}{{e}^{-j\cdot \Psi_{s_j}(\phi )}}.
\end{aligned}
\end{equation}

Next, the field coupling matrix $\mathbf{C}$ of the antenna array will be calculated. First, $M$ antennas are placed at \#1, \#2, $\cdots$, \#$M$ on the antenna mount respectively.  The antenna selection vector $s_n$ (where the $n$-th element is 1 and the other elements are 0) is applied  to excite the antenna array. To measure the amplitude pattern $\Lambda_{c_n}$ and phase pattern $\Psi_{c_n}$ of the $n$-th antenna in the array, the antenna selection vector is set to $s_n$.
Note that the measured amplitude pattern $P_r(\theta,\phi)$ is the power density of electric field $E(\theta,\phi)$, and the relation between them is
\vspace{-0.1cm}
\begin{align}\label{PnE}
P_r(\theta,\phi)=\frac{1}{2} \Re e\left[\mathbf{E} \times \mathbf{H}^*\right]= \frac{1}{2\eta} \left|E(\theta,\phi)\right|^2.
\end{align}
Thus, the information of the radiated electric field can be extracted by  \eqref{PnE}, the measured amplitude pattern, and the  phase pattern. The electric field of the $n$-th antenna thus is  
\begin{align}
    \mathbf{E}_{c_n}(\phi) = 2\eta\sqrt{\Lambda_{c_n}(\phi)}e^{j\cdot \Psi_{c_n}(\phi)}.
\end{align}

As a result, the field coupling matrix and the superdirective beamforming vector are calculated by \eqref{Cmatrix} and \eqref{mutuala}. Moreover, the non-ideal factors, that affect the antenna radiation pattern, have been considered in this practical measurement scheme.  It is noticeable that the coupling matrices are parameters related to the geometry structure of the antenna array. Thus when the array is fabricated, we only need to measure these parameters once.  

In our measurement-based approach, specific factors such as manufacturing tolerances and cable effects are carefully considered, as these factors can significantly alter the antenna radiation field.

\textbf{Manufacturing tolerances}:
\begin{itemize}
    \item \textbf{Impact}: Variations in the manufacturing process can lead to discrepancies in the physical dimensions and electrical properties of antenna elements, resulting in deviations from the intended radiation pattern.
    \item \textbf{Addressing the impact}: In our approach, we measure the actual radiation field of the antenna array, which inherently includes the effects of these manufacturing variations. By capturing the real-world performance of each antenna element, we can optimize the beamforming vectors to account for these discrepancies, ensuring that the final radiation pattern closely matches the desired design despite manufacturing tolerances.
\end{itemize}

\textbf{Cable effects}:
\begin{itemize}
    \item \textbf{Impact}: Cables can introduce phase shifts, losses, and impedance mismatches that affect the signals fed to each antenna element, leading to deviations in the expected radiation pattern.
    \item \textbf{Addressing the impact:} To address these issues, we employ a beamforming control board that is specifically designed to calibrate the antenna array. This control board allows us to adjust the amplitude and phase of the signals at each antenna element, compensating for any discrepancies introduced by the cables.
\end{itemize}

By incorporating these real-world factors into our measurement process, we are able to derive beamforming vectors that are tailored to the actual performance of the antenna array. This ensures that the final radiation field produced by the optimized beamforming vectors takes into account all relevant factors, including manufacturing tolerances and cable effects, leading to a more accurate and effective beamforming solution.
\section{Numerical Results}\label{IV}
In order to validate the effectiveness of the proposed superdirective beamforming methods, full-wave simulations and realistic experiments are carried out in this section. We compared our method with Maximum Ratio Transmission (MRT) and traditional superdirective beamforming methods\cite{altshuler2005monopole}\cite{ivrlavc2010toward}.

\subsection{Simulation Results}
 A printed dipole antenna array with a center frequency of 1600 MHz is designed as shown in Fig. \ref{fig:1}, where the width and length of the dipole are $w=1\ mm$ and $l=71.5\ mm$, respectively. The dipole antenna is printed on the FR-4 substrate ($\epsilon_r = 4.47$, $\mu_r=1$ , $\tan\delta= 0.0027$ and the thickness is $0.8\ mm$) of size $12.2\  mm\times78 \ mm $.  The transmit power is 1 W in simulations. The full-wave simulations are conducted in CST Microwave Studio.
\begin{figure}[htbp]
  \centering
  \includegraphics[width=2in]{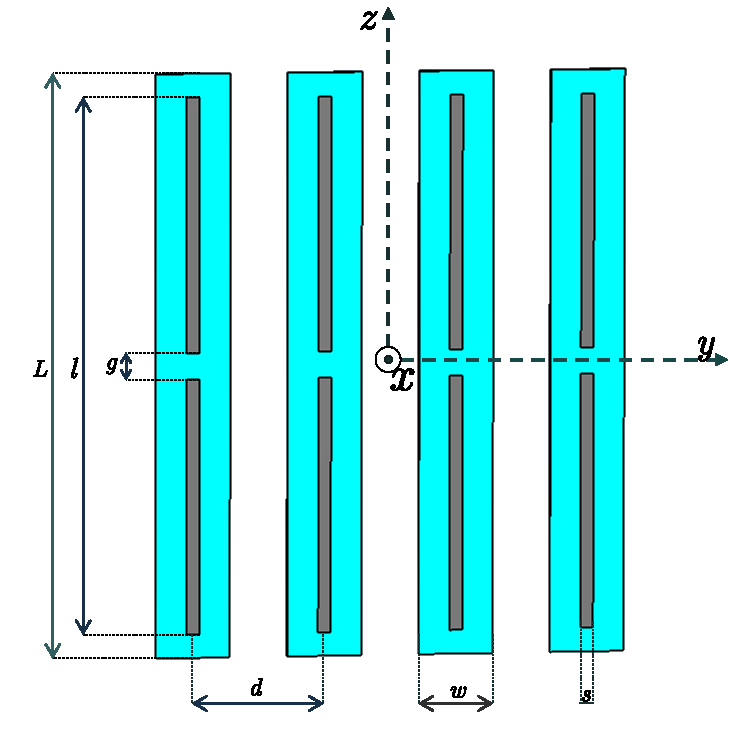}
  \caption{The schematic view of the printed dipole antenna array.}\label{fig:1}
\end{figure}

In conventional precoding strategies that ignore impedance coupling, MRT is considered to be the optimal beamforming vector that maximizes signal-to-noise ratio (SNR)\cite{MRT}. Since the impedance coupling is not considered, the maximum directivity factor of the MRT
 is proportional to the number of antennas $M$. However, by exploiting the impedance coupling effect between antennas, the maximum directivity factor can be much greater. In the following simulations, we will compare the performance of our proposed superdirectivity beamforming scheme and the MRT.

\addtolength{\topmargin}{0.01in}
The full-wave simulation results of the directivity as a function of the antenna spacing  are illustrated in Fig. \ref{fig:2}, for the case of $M=4$. In this figure,  The theoretical maximum directivity is obtained by simulating a single antenna and obtaining its radiation pattern. We then calculate the maximum directivity using equation \eqref{maxD}. In other words, the theoretical maximum directivity is calculated using numerical software and represents the theoretically maximum achievable directivity of an array composed of this antenna. It can be found that the maximum theoretical directivity increases with the decreasing of the antenna spacing, and the directivity reaches $17.8$ when the spacing is $0.15\lambda$. When the antenna array is directly excited by the beamforming vector calculated by \eqref{aa}, the directivity factor also increases at first as the spacing decreases, yet reaches a maximum around $0.28\lambda$. Afterwards, it starts to decrease with smaller spacing. This is because the smaller the spacing, the greater the field coupling between the antennas and the more pronounced the side effects of the traditional method, which does not take the field coupling into account. However, our proposed double-coupling based method has obvious gains and the directivity of the antenna array is in line with the theoretical value. 
The maximum directivity of MRT is obtained after the array is excited by $\mathbf{e}^*$ calculated by \eqref{vector_e} in full-wave simulation software, and 
It can be seen that the directivity of the MRT is consistently smaller than that of the proposed method or even the traditional method for antenna spacing less than $0.5\lambda$. When the spacing is around $0.5\lambda$, it can be found that the four curves almost overlap, which is because the coupling effects can be ignored under such antenna spacing.

\begin{figure}[htbp]
  \centering
  \includegraphics[width=3.5in]{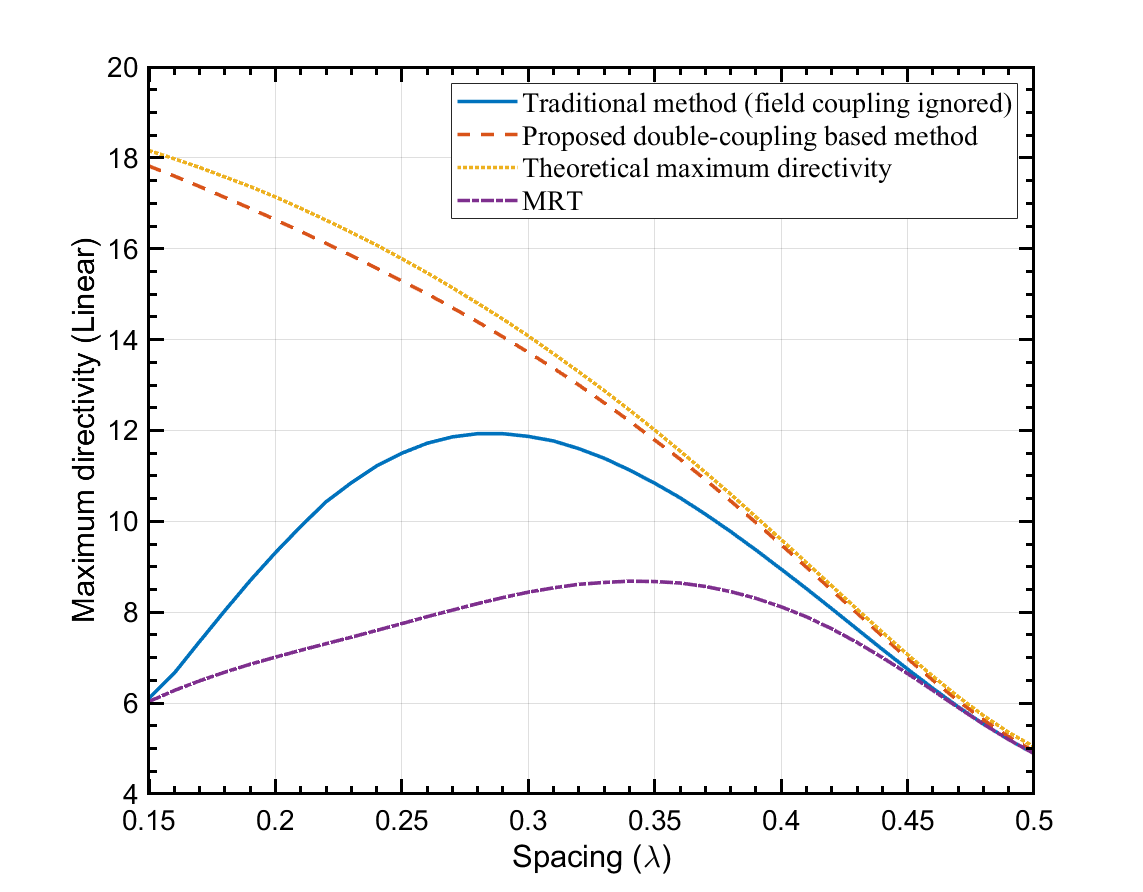}  \caption{The maximum directivity of a four-element dipole antenna array.}\label{fig:2}
\end{figure}
\begin{figure}[htbp]
  \centering
  \includegraphics[width=3.3in]{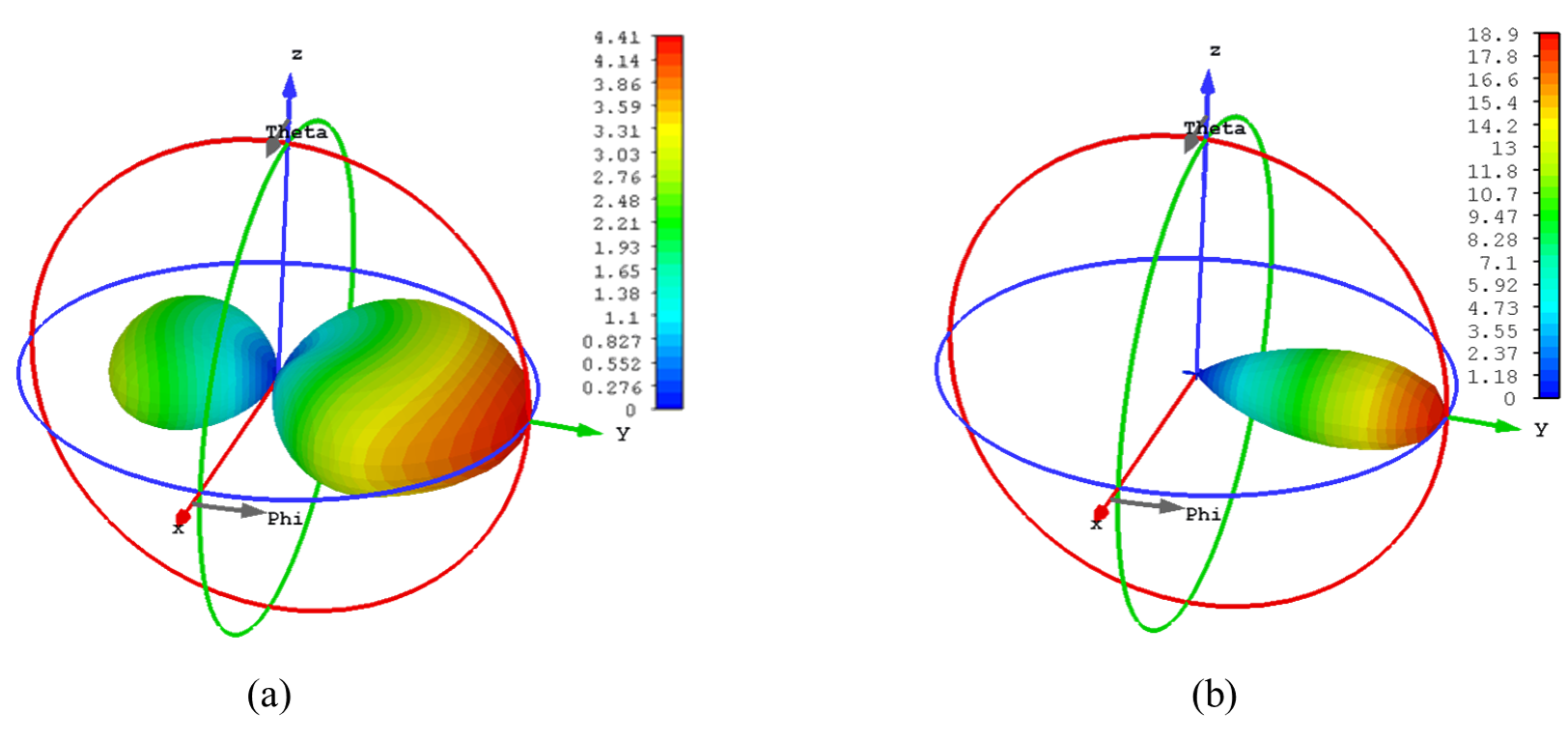}
  \caption{The radiation patterns of the linear dipole antenna arrays, where $M=4$, $d=0.1\lambda$. (a) Excited by the traditional superdirective beamforming method. (b) Excited by the proposed method.}\label{fig:4}\label{3D}
\end{figure}

Fig. \ref{3D} shows the 3D directivity pattern of the linear dipole array with four antennas. The antenna spacing is 0.1$\lambda$. Fig. \ref{3D} (a) and Fig. \ref{3D} (b) are the patterns when the beamforming vector is calculated by the traditional method that ignores the field coupling and by our proposed double-coupling based method respectively. Compared with Fig. \ref{3D} (a), it is evident that the radiated pattern is narrower and the back lobe is smaller in  Fig. \ref{3D} (b). 
It can be seen that the directivity factor of the traditional method is only 3.4, while the theoretical value is 18.9, which is much larger. 
When the proposed beamforming method is applied, the directivity is close to the theoretical value.

Fig. \ref{8antennas} shows the directivities for a eight-dipole antenna array. 
It can be seen that the directivity based on the traditional method has a poor performance in the small spacing region, while our proposed method is close to the theoretical values. 
 The performance degradation of the traditional method at smaller antenna spacings is primarily due to the increased electromagnetic field coupling between the antennas. As the spacing between antennas decreases, the mutual coupling becomes more significant. This mutual coupling is represented by the coupling matrix \(\mathbf{C}\). When the antennas are closely spaced, \(\mathbf{C}\) deviates more from the identity matrix, indicating stronger interactions between the antennas.
These deviations impact both the amplitude and phase of the beamforming vectors, which are critical in shaping the desired beam pattern. As a result, the intended beam pattern is distorted, leading to performance degradation. In essence, the closer the antennas are, the more they interfere with each other, which undermines the effectiveness of the beamforming process.

In both cases, the proposed method, which accounts for mutual coupling effects, significantly outperforms the MRT method, especially as the antenna spacing decreases. This is evident from the maximum directivity (measured in linear units), where the MRT method's performance degrades more rapidly as antenna spacing decreases compared to the proposed method. The practical implications of these findings are significant. The MRT design does not consider the impact of impedance coupling that becomes prominent as antenna spacing decreases. This oversight leads to suboptimal performance compared to the superdirectivity method that incorporates coupling effects. Therefore, in practical compact antenna systems, especially where space constraints lead to reduced antenna spacing, it is crucial to account for mutual coupling in the system design to achieve optimal performance.

Fig. \ref{limited_angles_pattern} compares the directivity patterns for three different methods applied to an antenna array with 8 antennas spaced at 0.3$\lambda$ apart. The Proposed method (full field data), shown in solid blue, uses field data from all angles. The Proposed method (four angles data), represented by the dashed orange line, only uses field data from 4 specific angles to compute the field coupling matrix. Finally, the traditional method (dotted red line) ignores the field coupling effect, leading to notable differences in the directivity pattern. It can be observed that using only four angle values to compute the coupling matrix 
$\mathbf{C}$ achieves almost the same results as using full-field data. This validates the accuracy of our proposed Corollary \ref{coro1}.

Fig. \ref{deltaDdeltaP} demonstrates the relationship between antenna spacing and beamforming accuracy of a four-dipole array, highlighting how decreasing the spacing between antennas enhances the field coupling effects. This increase in coupling leads to a greater deviation from the expected performance in beamforming.
The left y-axis shows \(\Delta D\), which represents the difference between the actual directivity and the theoretical directivity. As antenna spacing decreases, \(\Delta D\) increases, indicating that the actual directivity becomes less accurate compared to the theoretical model.
The right y-axis displays \(\Delta F\), representing the logarithmic mean squared error (MSE) between the theoretical and actual radiation patterns. Similar to \(\Delta D\), \(\Delta F\) increases as the antenna spacing decreases, reflecting greater inaccuracies in the radiation pattern due to the stronger mutual coupling effects. This illustrates that as the antennas are placed closer together, the beamforming precision deteriorates, resulting in higher discrepancies from the desired radiation characteristics.

A table summarizing the key performance metrics, including beamwidth, peak side lobe level (PSLL), and the directivity is provided. Table \ref{tab:comparison} summarizes the performance metrics for each method and configuration, facilitating an easier comparison and highlighting the advantages of the proposed method.

\begin{table*}[h!]
    \centering
    \caption{Comparison of Beamwidth, PSLL, and Directivity for Different Methods and Antenna Configurations}
    \centering
    \begin{tabular}{lccc|ccc}
        \toprule
        \textbf{Method} & \multicolumn{3}{c|}{\textbf{4 Antennas (0.3$\lambda$)}} & \multicolumn{3}{c}{\textbf{8 Antennas (0.2$\lambda$)}} \\
        \midrule
        & \textbf{Beamwidth} & \textbf{PSLL} & \textbf{Directivity} & \textbf{Beamwidth} & \textbf{PSLL} & \textbf{Directivity} \\
        \midrule
        \textbf{MRT} & 83$^\circ$ & -17.6 dB & 8.7 & 58.4$^\circ$ & -8.9 dB & 11.7 \\
        \textbf{Traditional} & 51.1$^\circ$ & -4.6 dB & 11.6 & 87$^\circ$ & -6 dB & 6.1 \\
        \textbf{Proposed} & 55.8$^\circ$ & -8.48 dB & 14 & 24.5$^\circ$ & -16.2 dB & 61.58 \\
        \bottomrule
    \end{tabular}
    \label{tab:comparison}
\end{table*}

\begin{figure}[htbp]
  \centering
  \includegraphics[width=3.5in]{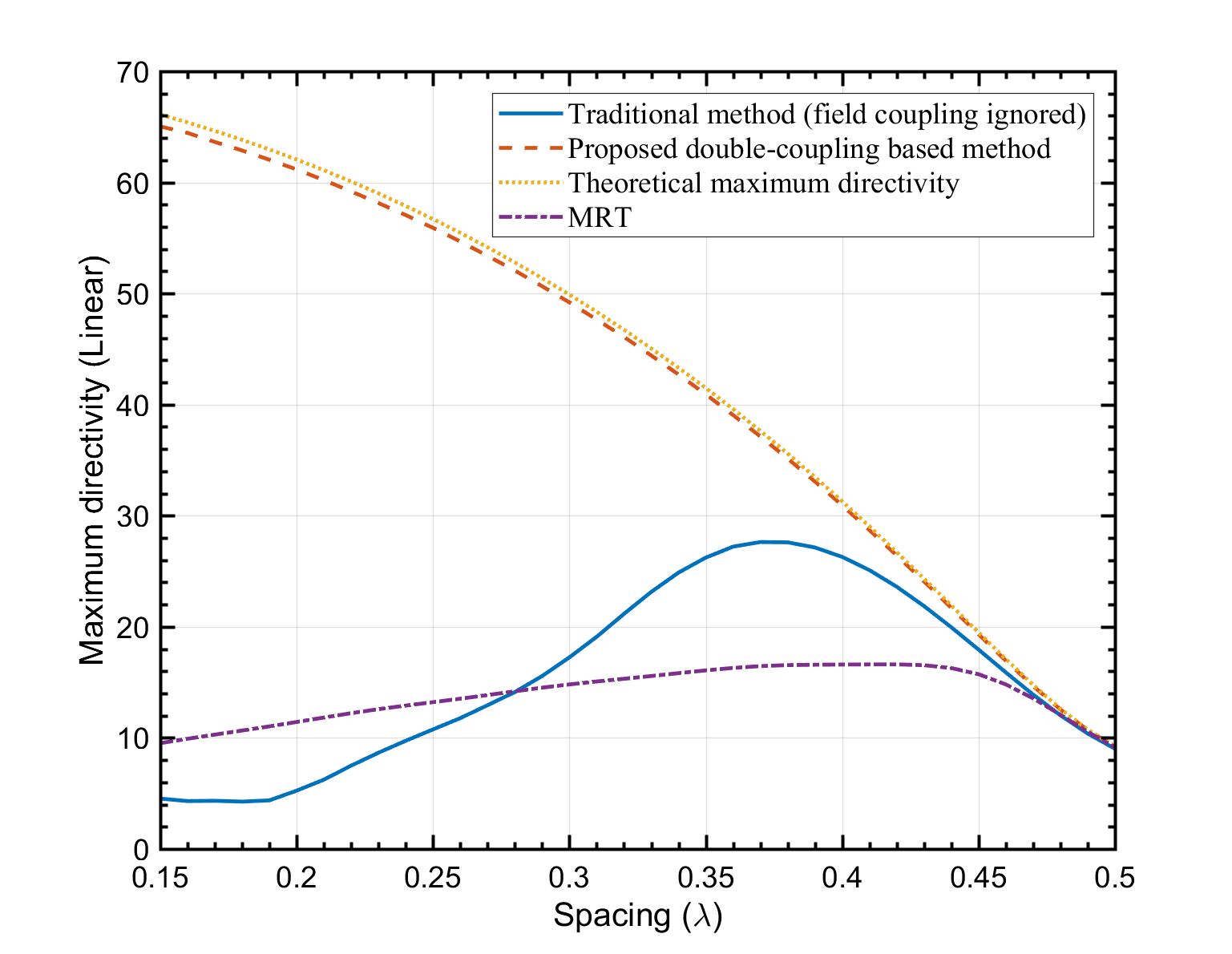}
  \caption{The maximum directivity of an eight-element dipole antenna array.}\label{8antennas}
\end{figure}

\begin{figure}[htbp]
  \centering
  \includegraphics[width=3.5in]{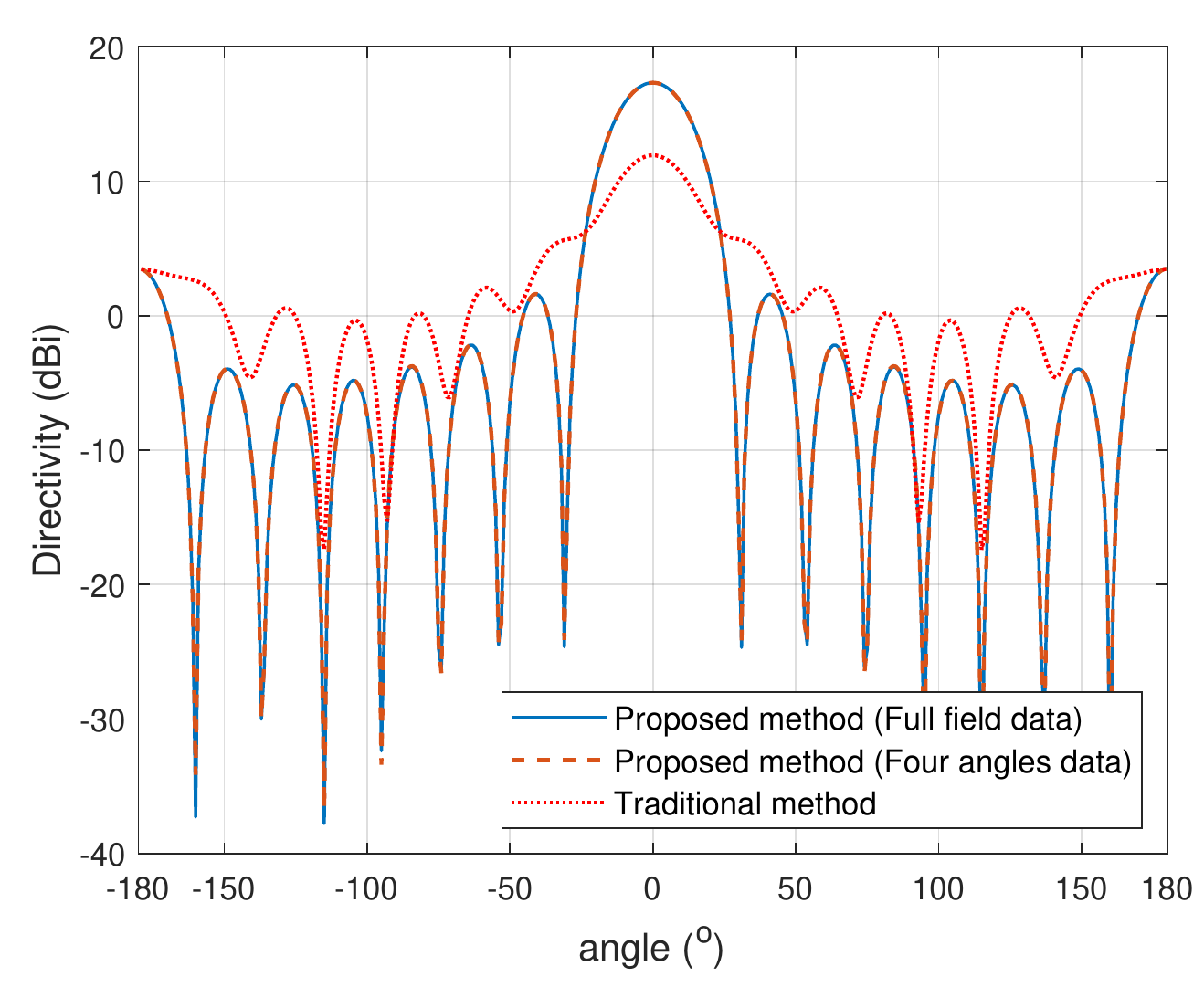}
  \caption{Directivity patterns comparison for different methods with 8 antennas (0.3$\lambda$ Spacing)}\label{limited_angles_pattern}
\end{figure}

\begin{figure}[htbp]
  \centering
  \includegraphics[width=3.2in]{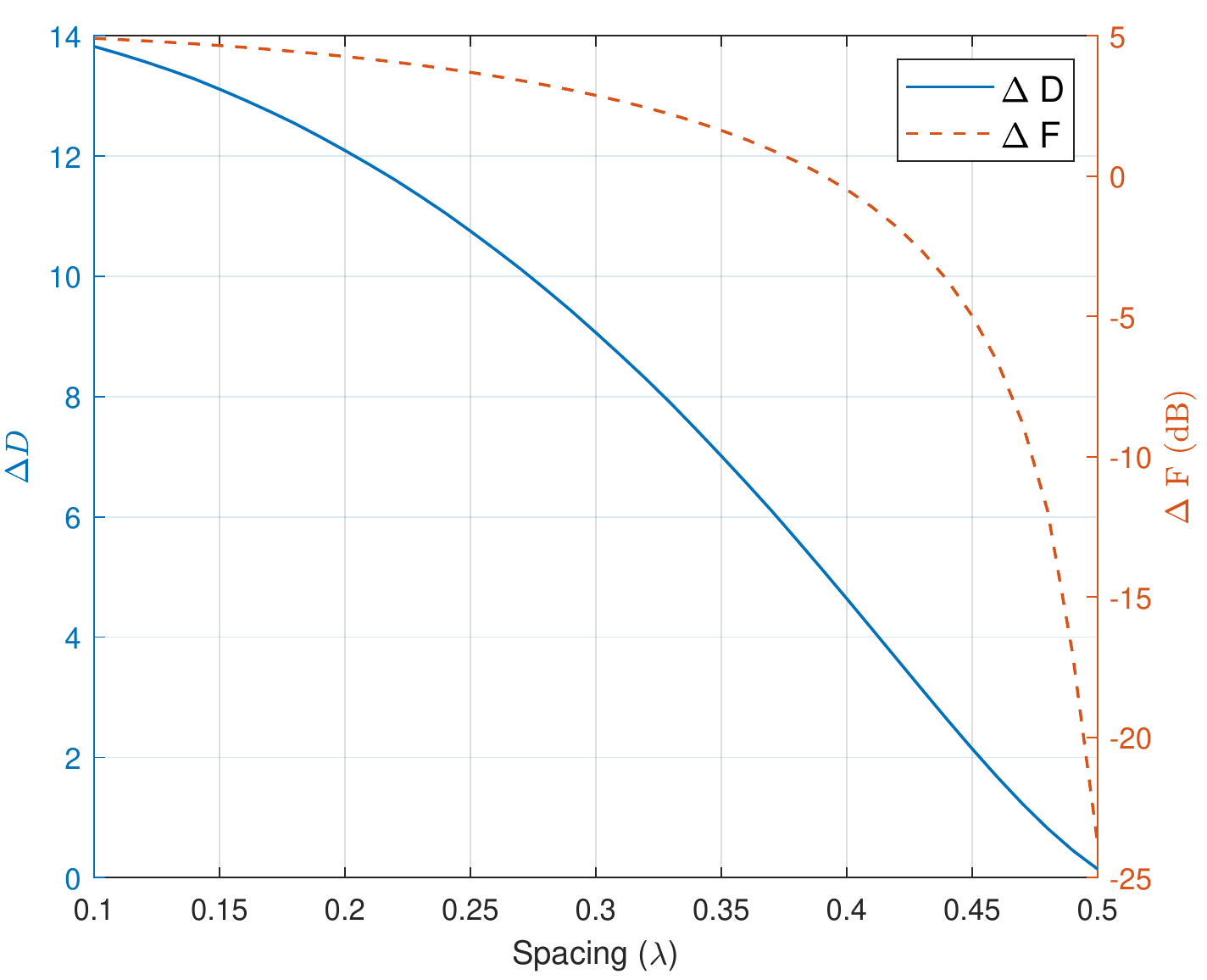}
  \caption{Increasing beamforming inaccuracy due to enhanced field coupling with decreasing antenna spacing}\label{deltaDdeltaP}
\end{figure}
To analyze the impact of ohmic loss on the performance of directivity, the gain of the antenna array with different spacing is illustrated in Fig. \ref{fig:5}. When the number of antennas is 2 and 4 respectively,  the radiation efficiency of the antenna is 96\%. It can be found the gain does not keep increasing as the antenna spacing decreases, and the larger the number of antennas, the more obvious the effect of ohmic loss at smaller spacing. Consequently, the maximum gain is 9.6 when the number of antennas is 4 and the spacing is 0.33$\lambda$. This phenomenon can be explained as follows.
\begin{figure}[tbp]
  \centering
  \includegraphics[width=3in]{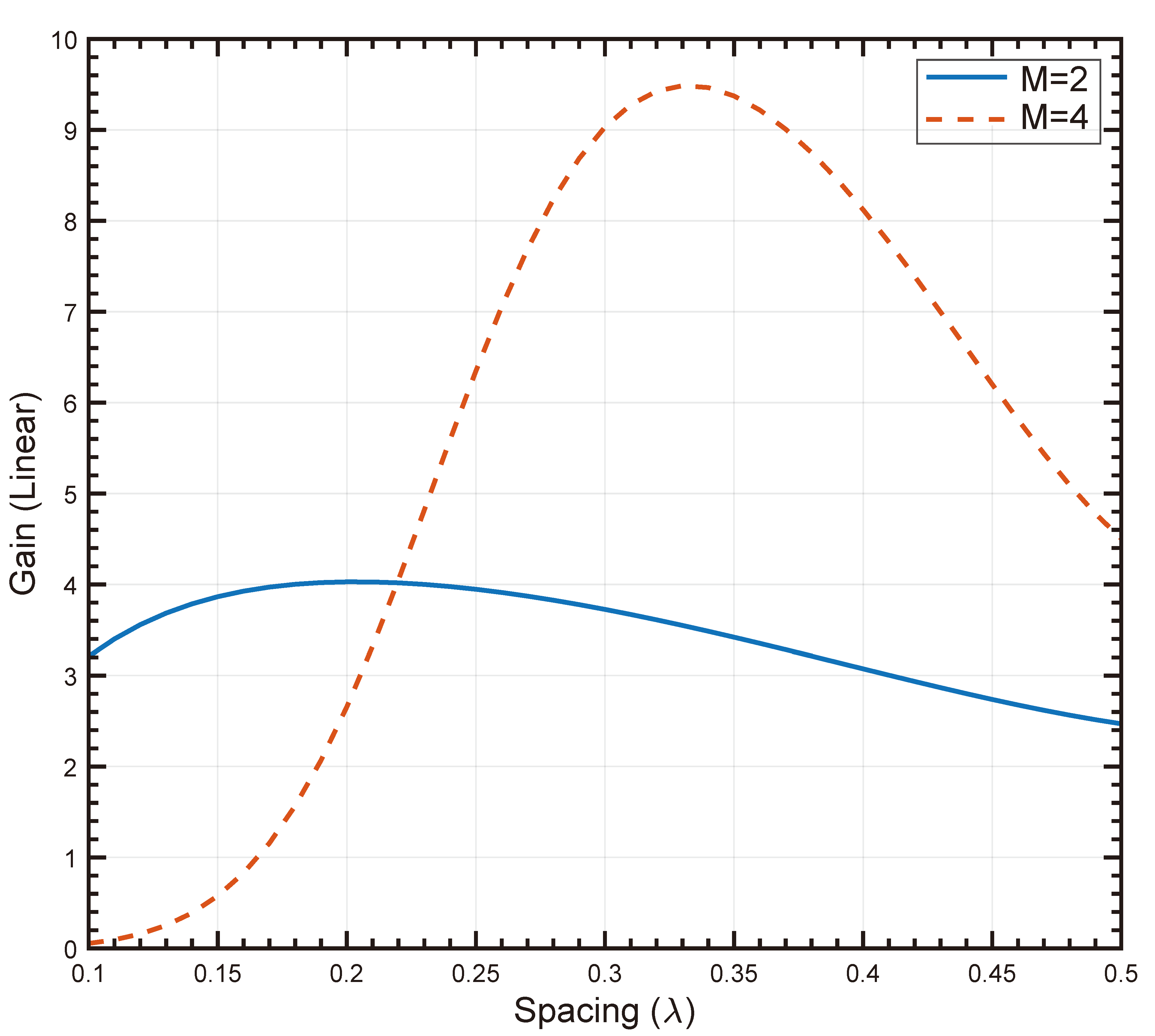}
  \caption{The gain achieved by the proposed method in the presence of the ohmic loss. }\label{fig:5}
\end{figure}

We analyze the expression \eqref{gain} for array gain considering ohmic losses. To simplify the notation, we define $\bar{\mathbf{b}}=\mathbf{Cb}$. Thus, \eqref{gain} can be rewritten as:
\begin{align}
    G=\frac{\bar{\mathbf{b}}^T \mathbf{e e}^H\bar{\mathbf{b}}^*}{\bar{\mathbf{b}}^T\left(\mathbf{Z}+r_{\text {loss }} \mathbf{I}_M\right)\bar{\mathbf{b}}^*}.
\end{align}

According to the physical meaning of $G$, the denominator of this expression can be split into two terms: the radiated power $P_{\rm rad}$ and the loss power $P_{\rm loss}$:
\begin{align}
P_{\rm rad} = \bar{\mathbf{b}}^T\mathbf{Z}\bar{\mathbf{b}}^*,
\end{align}
\begin{align}
P_{\rm loss}=\bar{\mathbf{b}}^Tr_{\rm loss}\mathbf{I}_M\bar{\mathbf{b}}^*.
\end{align}

Without loss of generality, we ignore the power constraint coefficient $\zeta$ for the superdirective beamforming vector, so that $\bar{\mathbf{b}}=\mathbf{Z}^{-1}\mathbf{e}^*$.
Since the impedance coupling matrix $\mathbf{Z}$ of the dipole antenna is a real symmetric matrix, it can be diagonalized as:
\begin{align}
\mathbf{Z} = \mathbf{U}\mathbf{\Lambda}\mathbf{U}^H,
\end{align}
where $\mathbf{U}$ is a unitary matrix. Therefore,
\begin{align}
P_{\rm rad} &= \mathbf{e}^H\mathbf{Z}^{-1}\mathbf{e}\\ \notag
&=(\mathbf{U}^H\mathbf{e})^H\mathbf{\Lambda}^{-1}(\mathbf{U}^H\mathbf{e})\\\notag
&=\sum_{i=1}^M\frac{1}{\lambda_i}\mathbf{w}^H[i]\mathbf{w}[i],\notag
\end{align}
\begin{align}
P_{\rm loss} &= r_{\rm loss}\mathbf{e}^H\mathbf{Z}^{-1}\mathbf{Z}^{-1}\mathbf{e}\\\notag
&=(\mathbf{U}^H\mathbf{e})^Hr_{\rm loss}\mathbf{\Lambda}^{-2}(\mathbf{U}^H\mathbf{e})\\\notag
&=\sum_{i=1}^M\frac{r_{\rm loss}}{\lambda_i^2}\mathbf{w}^H[i]\mathbf{w}[i],\notag
\end{align}
where $\mathbf{w}=\mathbf{U}^H\mathbf{e}$ and $\mathbf{w}[i]$ is the $i$-th element of $\mathbf{w}$. When the spacing between antennas is relatively large (around $0.5\lambda$), the impedance matrix $\mathbf{Z}$ tends to approach the identity matrix. In this scenario, with no small eigenvalues in $\mathbf{Z}$ and low ohmic loss $r_{\rm loss}$, the gain gradually increases as the spacing decreases. However, this trend stops when the spacing is reduced to a certain point. As the antenna spacing decreases, the impedance coupling matrix $\mathbf{Z}$ tends to become singular, i.e., some of its eigenvalues approach zero, causing the magnitude of $P_{\rm loss}$ to become comparable to or even greater than that of $P_{\rm rad}$, resulting in a decrease in gain.
\begin{table}[]
\caption{Directivity and realized gain with matched antenna array}
\centering
\label{dir_rg}
\begin{tabular}{ccc}
\hline
Array configuration & Directivity & Realized gain \\ \hline
$M=4,d=0.3\lambda$  & 14          & 13.9          \\ \hline
$M=4,d=0.2\lambda$  & 16.6        & 16            \\ \hline
\end{tabular}
\end{table}

Several practical methods can be suggested to mitigate ohmic losses in real-world implementations, especially in the context of the proposed beamforming system:

\textbf{ Advanced materials:}
\begin{itemize}
    \item  \textbf{Superconductors}: Recently, superconducting materials have gained attention due to their near-zero electrical resistance at specific low temperatures. Using superconductors to fabricate antennas and transmission lines can significantly reduce ohmic losses. This reduction is due to the elimination of resistive heating, leading to improved efficiency and power handling capabilities in the beamforming system.
    \item \textbf{Low-resistivity alloys}: Since superconductors are not feasible currently, using low-resistivity materials such as silver, copper alloys, or graphene could reduce the resistive losses compared to traditional copper. These materials can be strategically applied to critical components to enhance overall system efficiency.
\end{itemize}

\textbf{ Cooling techniques:}
\begin{itemize}
    \item \textbf{Cryogenic cooling}: By placing the system in a cryogenic environment, where temperatures are significantly lower, the resistivity of conductive materials decreases. This reduction in resistivity directly translates to lower ohmic losses. 
\end{itemize}

\textbf{ Integration into the proposed beamforming system:}
\begin{itemize}
    \item \textbf{Superconducting components}: Superconducting materials could be integrated into the antennas and critical transmission lines of the beamforming system. This integration would involve replacing conventional conductors with superconductors. This could lead to a significant reduction in power consumption and an increase in signal strength and clarity.
    \item \textbf{Cryogenic environment}: The entire beamforming array could be housed within a cryogenic environment, ensuring that all components operate at reduced resistivity. This approach would require advanced thermal management and insulation techniques to maintain the low temperatures required for superconductivity or to achieve significant reductions in resistivity in non-superconducting materials.
\end{itemize}

 In practice, the realized gain is also an important parameter of the antenna array that accounts for both the ohmic loss and mismatch between sources and antennas. In a compact antenna array, the input impedance of antennas can change due to strong mutual coupling, leading to a decrease in the realized gain. This issue can be mitigated by designing an appropriate matching network. To demonstrate the efficacy of the impedance matching, we here consider the perfect electric conductor (PEC) as the material of antennas to exclude the  ohmic loss in simulations.
Table \ref{dir_rg} presents the directivity and realized gain of a matched antenna array obtained through electromagnetic simulation. It is observed that the realized gain can closely approach the directivity. The design of the matching networks is out of the scope of this paper, and will be studied in our future works. 
\subsection{Experimental Results}

\begin{figure}[htbp]
  \centering
  \includegraphics[width=2.6in]{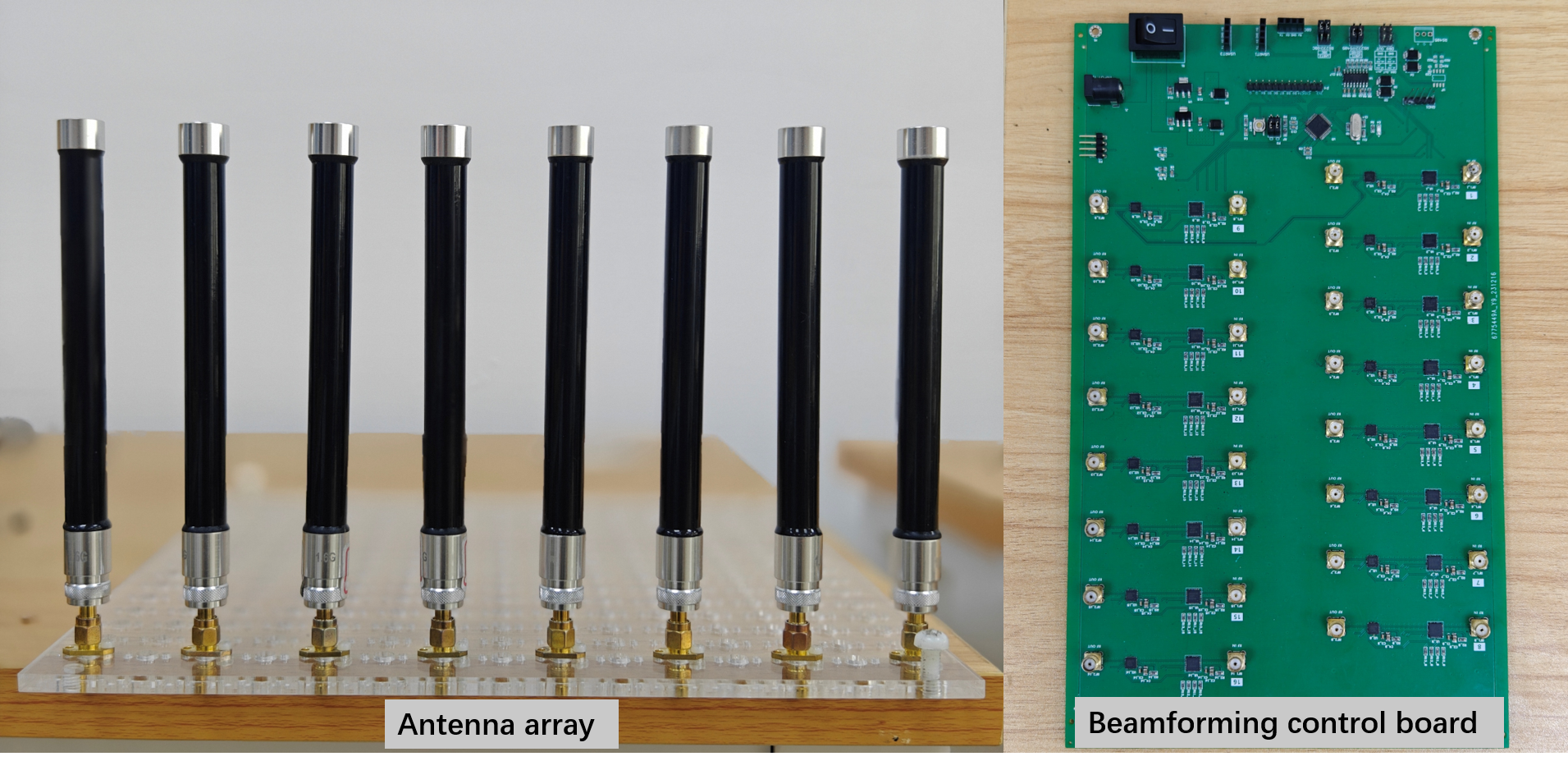}  \caption{The  antenna array and beamforming control board. }\label{bfcb}
\end{figure}

\begin{figure}[htbp]
  \centering
  \includegraphics[width=2.6in]{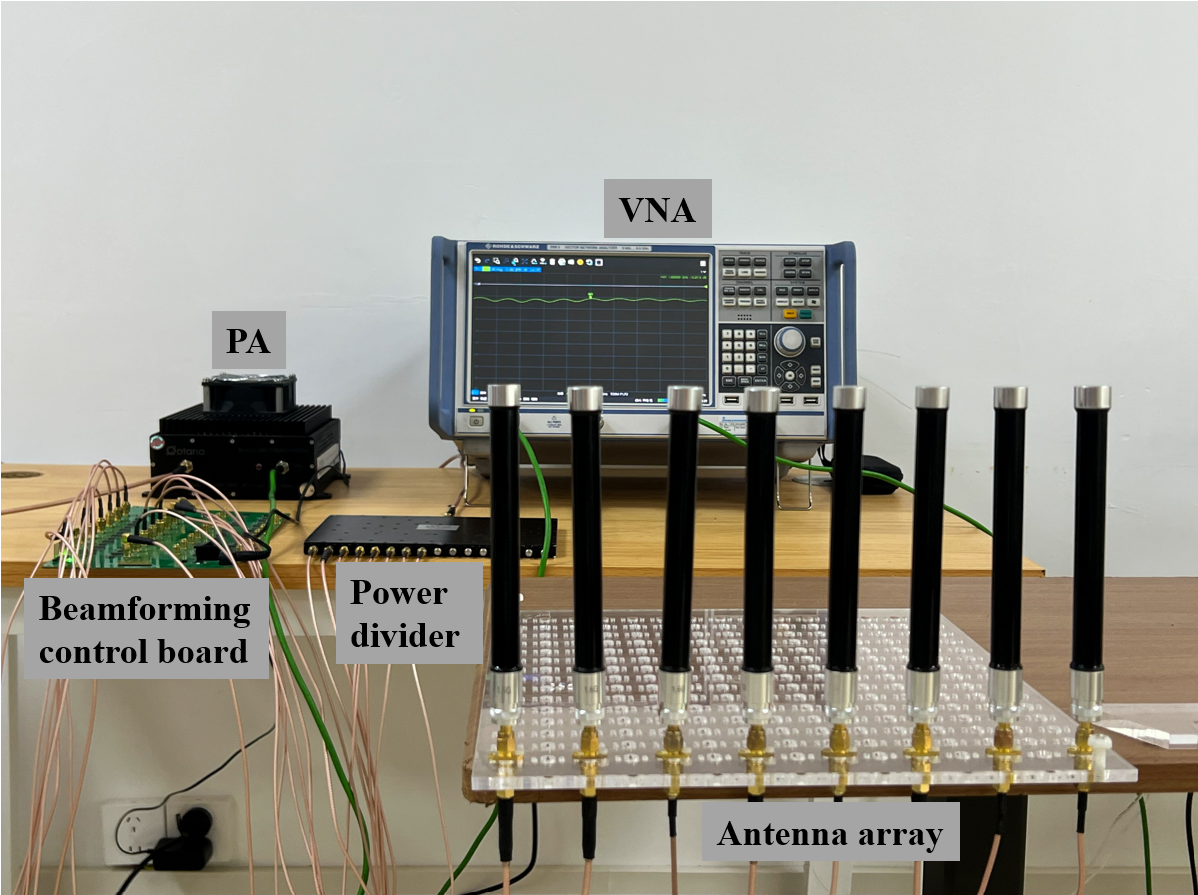}  \caption{The hardware platform of the superdirective antenna array testing system. }\label{tsd}
\end{figure}
 \begin{figure}[htbp]
  \centering
  \includegraphics[width=2.6in]{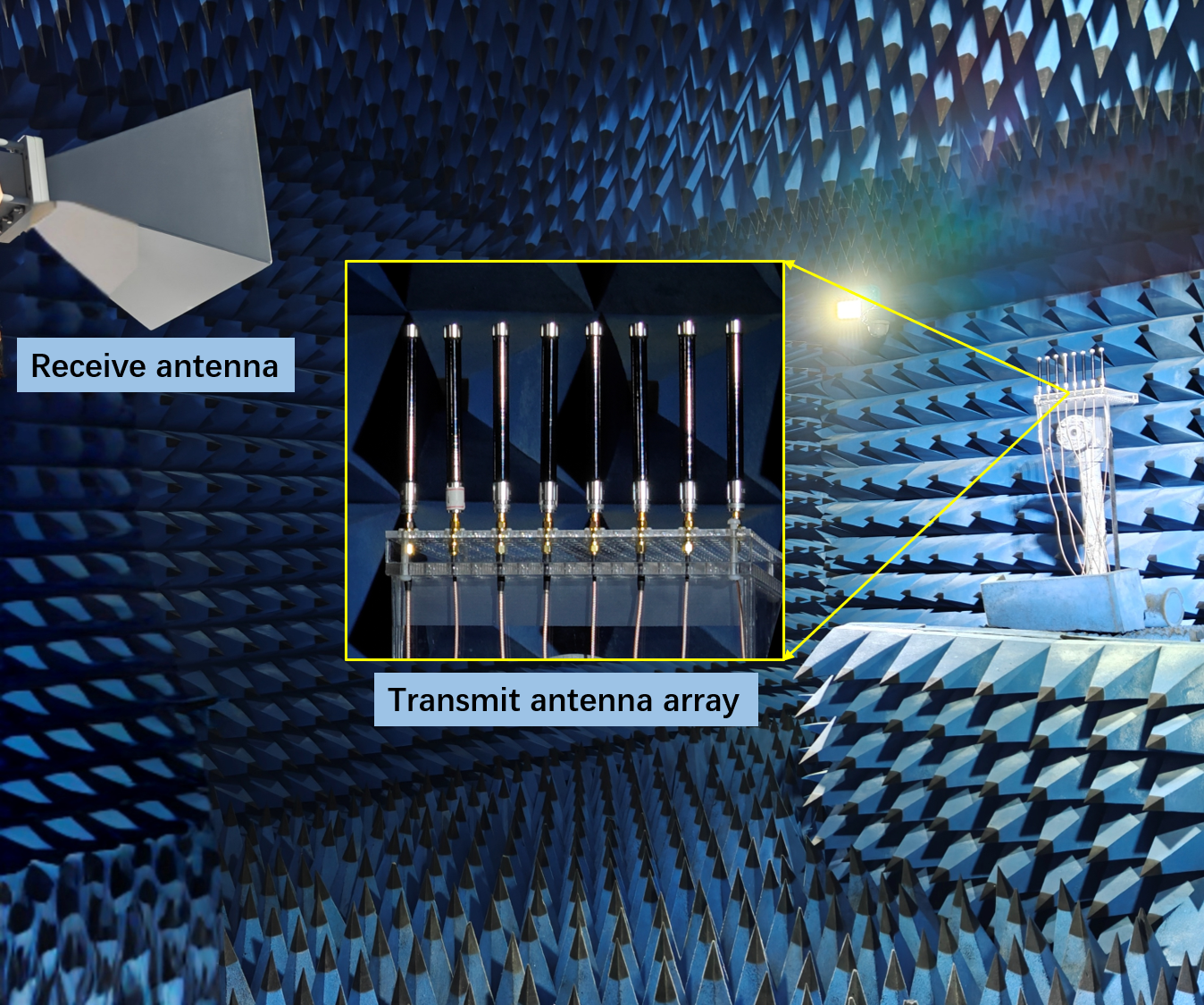}  \caption{The microwave anechoic chamber of size 4 m × 6 m × 4 m. The Array-RX antenna distance is 4.8 m. }\label{fig:6}
\end{figure}
To validate the proposed beamforming method in practice, we design and fabricate a prototype of the superdirective antenna array. 
The dipole antenna is selected as the radiating element of the array, and its parameters are consistent with the simulation settings. Each antenna is connected to a coaxial feed line with a SubMiniature version A (SMA) connector.

In order to excite each antenna with software-controlled phase and amplitude, we develop a beamforming control board that has 16 channels, and the beamforming coefficient for each channel is quantized to 8 bits in phase and 7 bits in amplitude.  
\begin{figure}[htbp]
  \centering
  \includegraphics[width=3in]{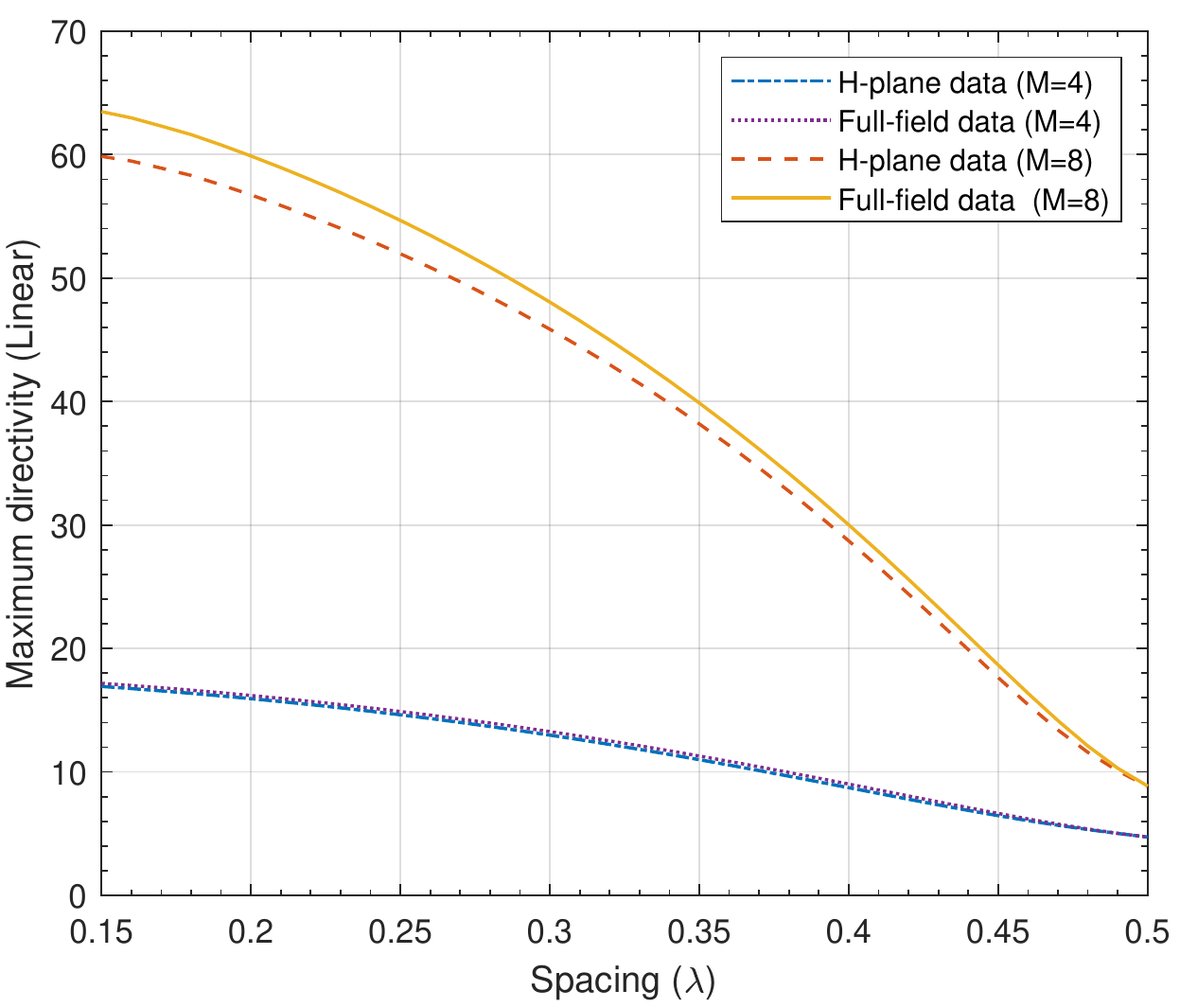}
  \caption{Comparison of beamforming performance using full-field data vs. H-Plane field data.}\label{differentZ}
\end{figure}
\begin{figure*}[htbp]
\centering
\begin{minipage}{\textwidth}
\centering
\subfigure[Comparing with the maximum ratio transmission. ]{
\includegraphics[width=5.3cm]{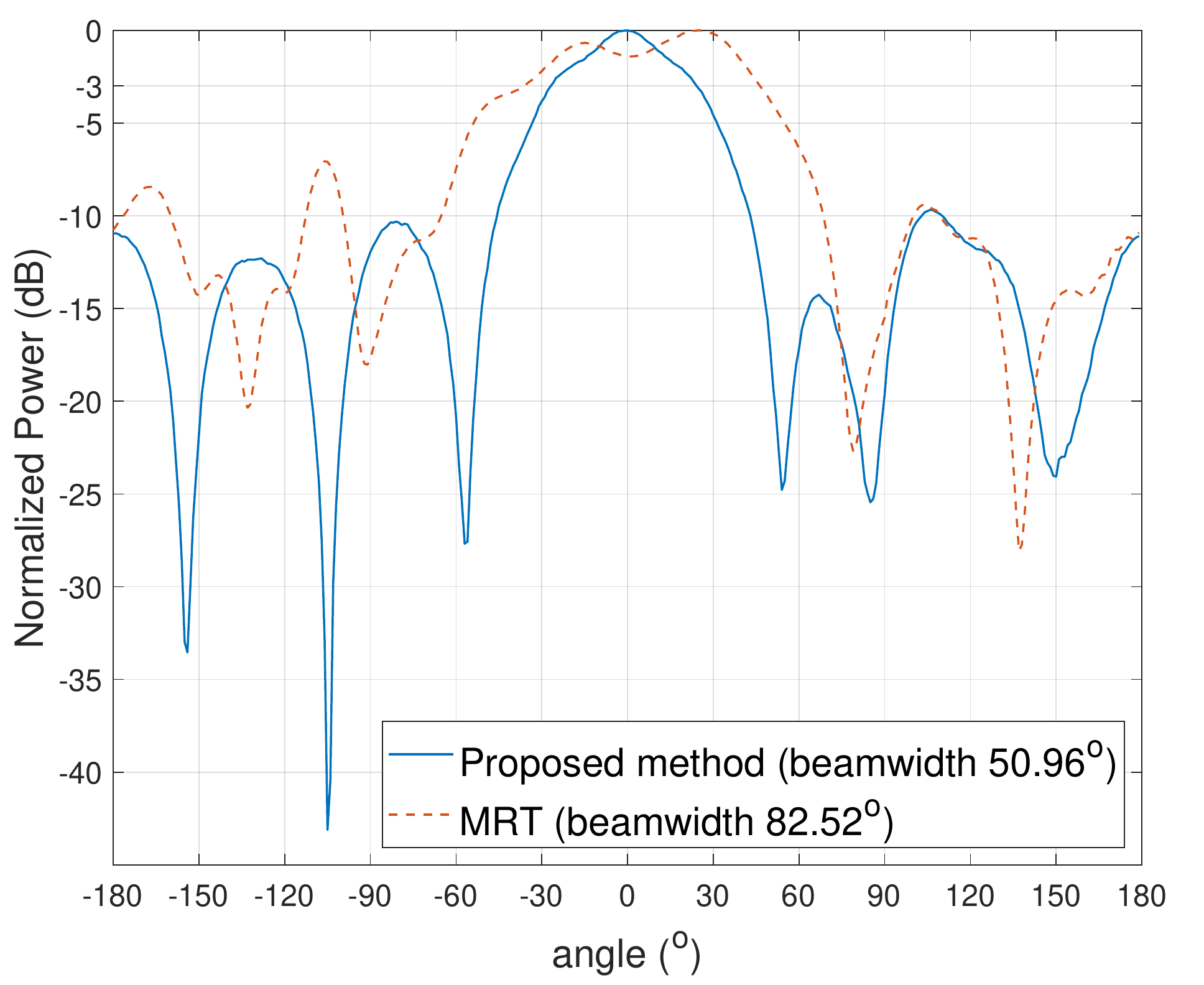}
}
\quad
\subfigure[Comparing with the traditional superdirective beamforming method.]{
\includegraphics[width=5.3cm]{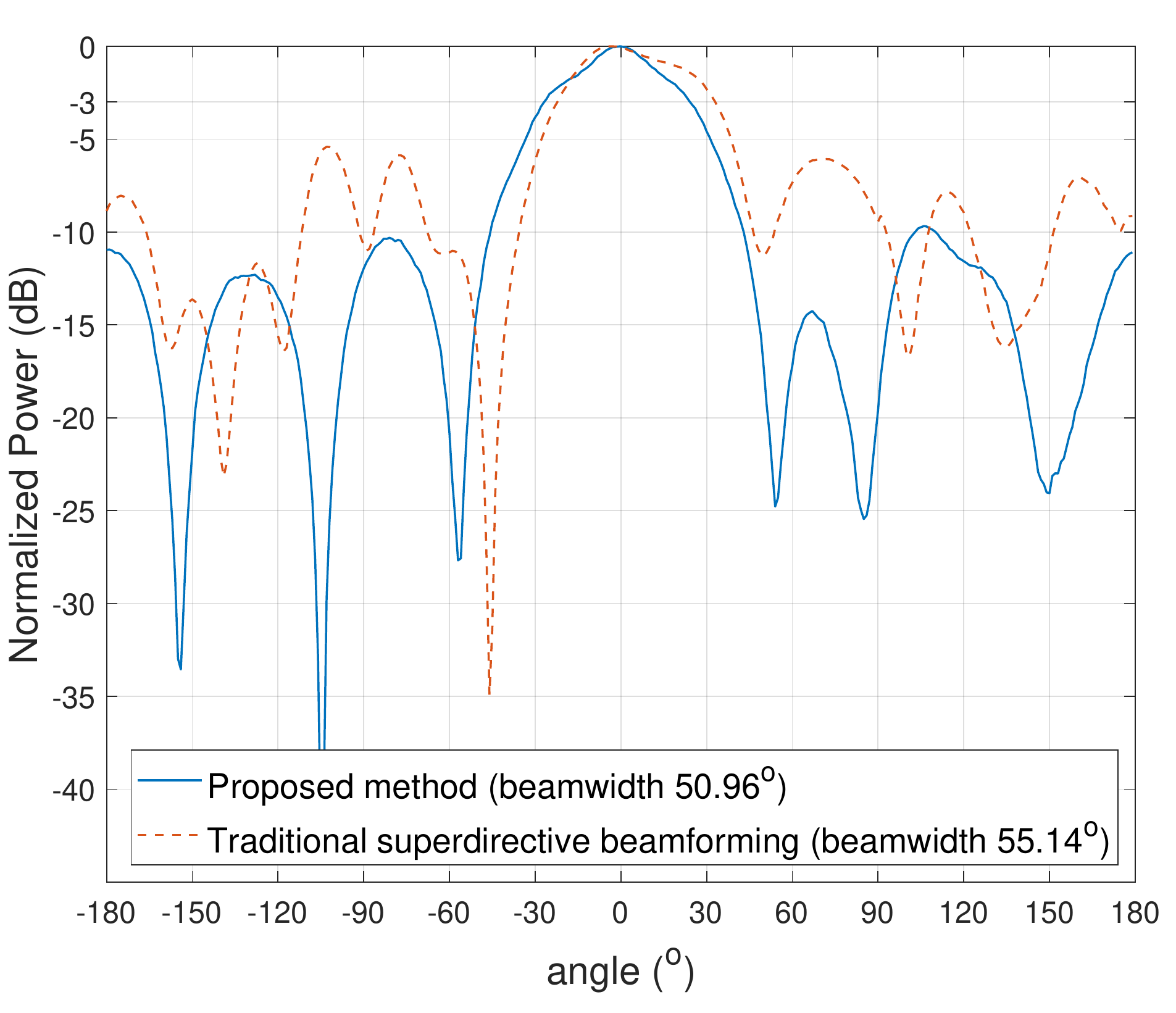}
}
\quad
\subfigure[Comparing with the simulation results.]{
\includegraphics[width=5.3cm]{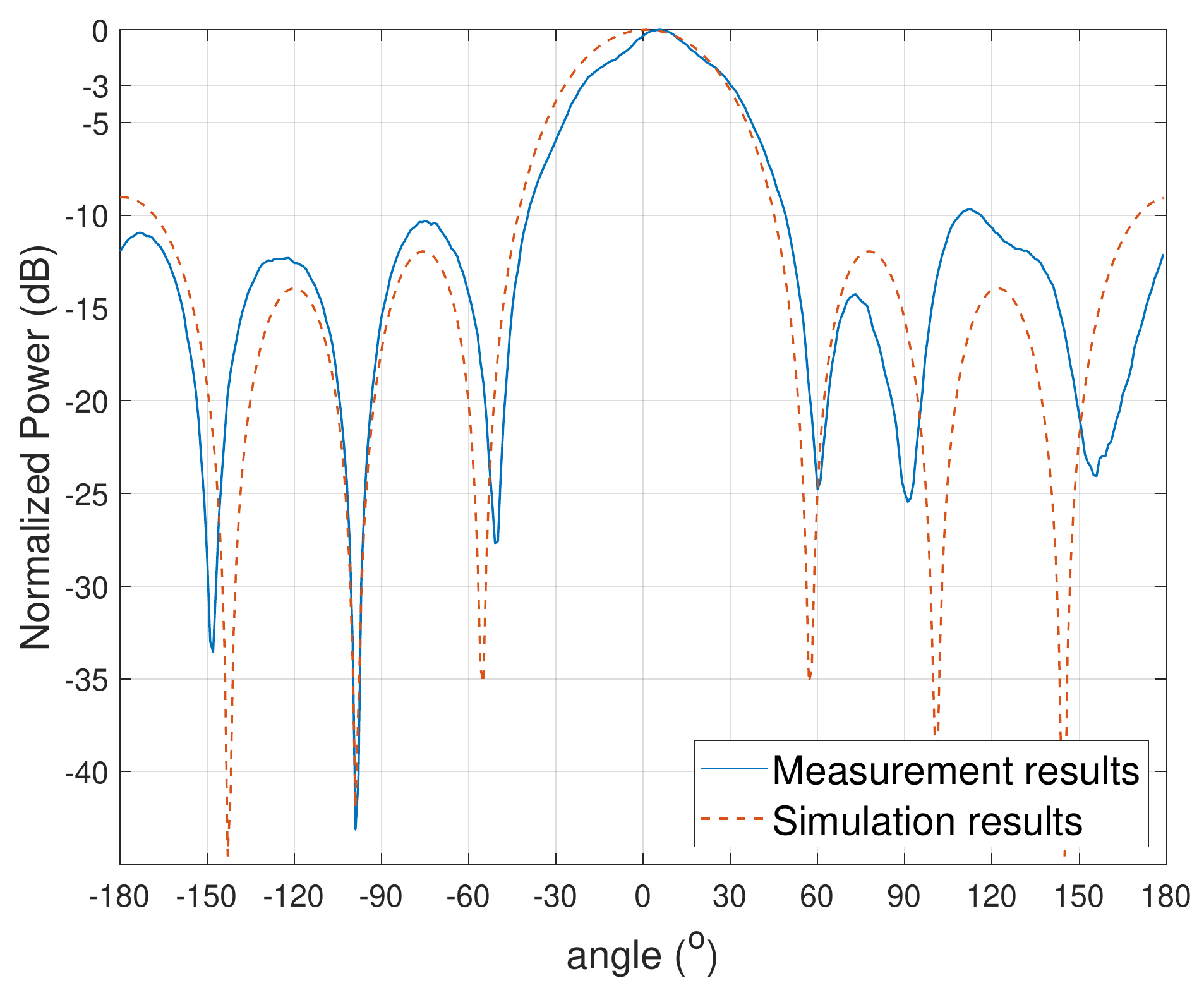}
}
\caption{The measured normalized power of the proposed beamforming method in principal radiation plane ($\theta=90^o$) comparing with (a) the MRT, (b) the traditional superdirective beamforming method (field coupling ignored) and (c) the full-wave simulation results. The number of antennas is 4 and the antenna spacing is 0.3$\lambda$.}\label{R2}
\end{minipage}
\begin{minipage}{\textwidth}
\centering
\subfigure[Comparing with the maximum ratio transmission. ]{
\includegraphics[width=5.3cm]{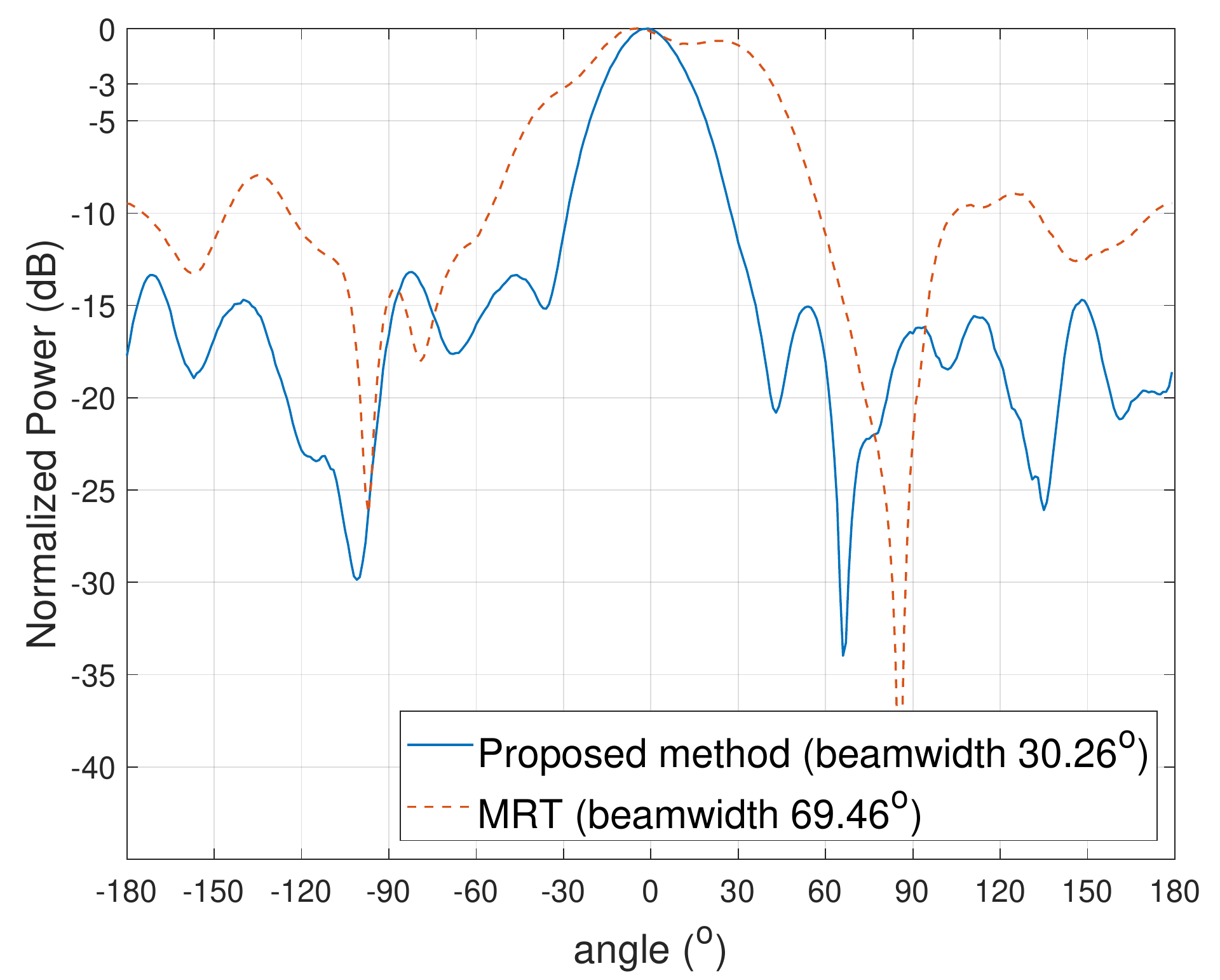}
}
\quad
\subfigure[Comparing with the traditional superdirective beamforming method.]{
\includegraphics[width=5.3cm]{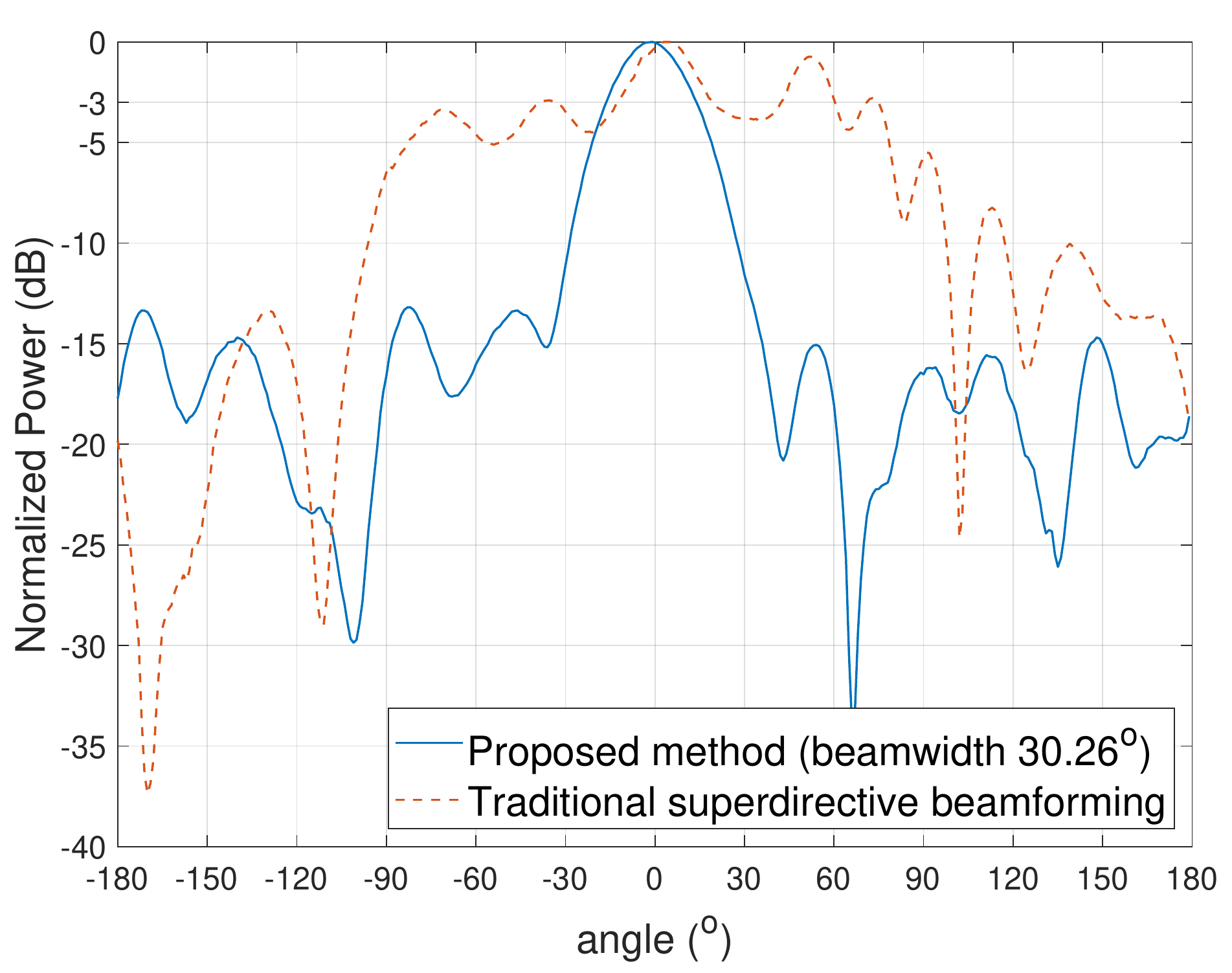}
}
\quad
\subfigure[Comparing with the simulation results.]{
\includegraphics[width=5.3cm]{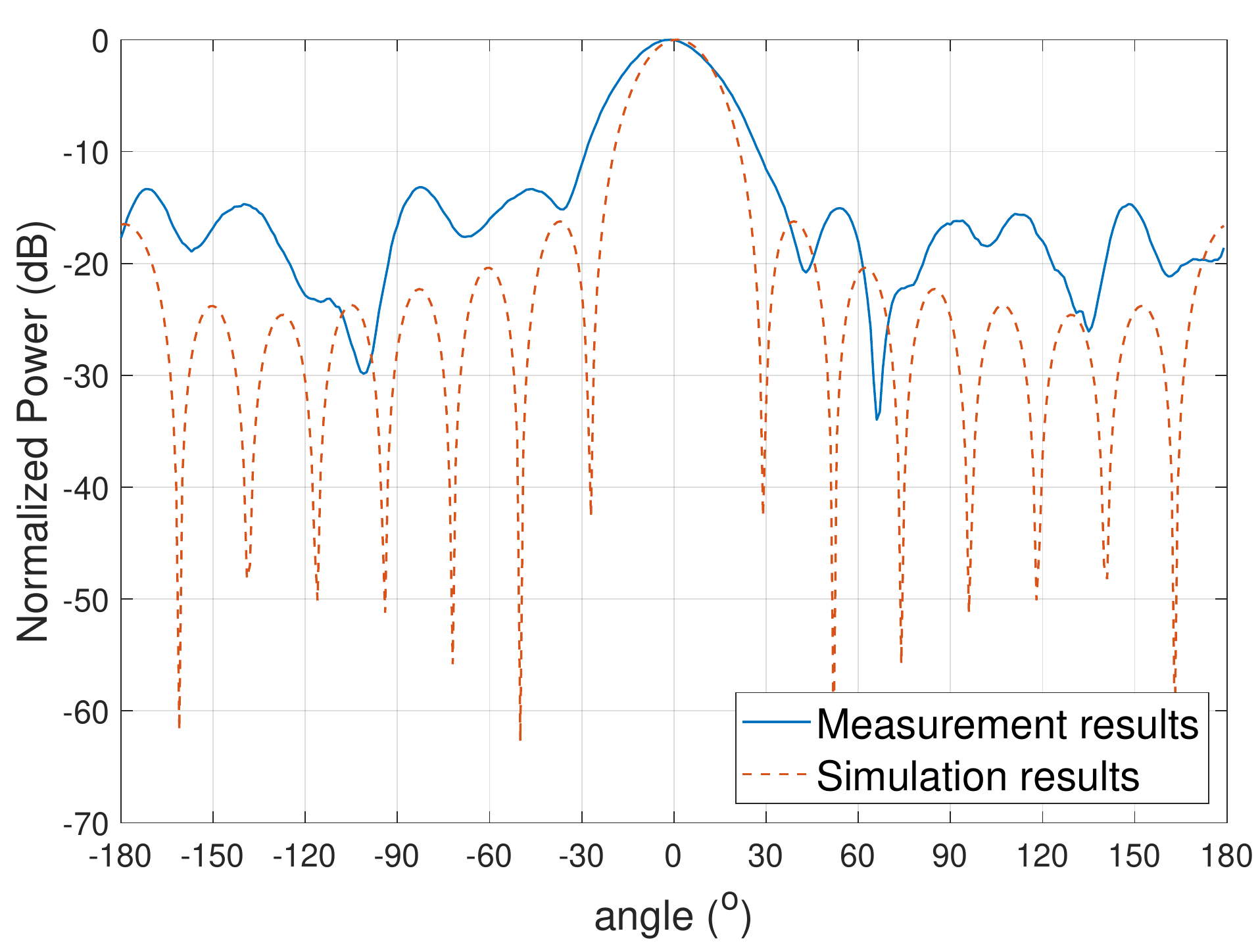}
}
\caption{The measured normalized power in the principal radiation plane ($\theta=90^o$) with 8 dipole antennas and an antenna spacing of 0.2$\lambda$.}\label{R53}
\end{minipage}
\end{figure*}
The  antenna array and beamforming control board are shown in Fig. \ref{bfcb}. A 16-way power divider is utilized to extend the RF signal and its output is connected to the input of the beamforming control board.  Note that some non-ideal factors are inevitable in the hardware platform, such as quantization errors for the beamforming vector. The performance loss due to excitation errors is discussed in \cite{gao2023robust}, wherein we demonstrate a method to enhance the robustness of superdirective beamforming to excitation errors. The antenna array to be measured is connected to the output of the beamforming control board. The initial phase differences between the channels are measured by the vector network analyzer (VNA) and calibrated by the beamforming control board. The diagram of the hardware system is shown in Fig. \ref{tsd}.

 The experiments are carried out in the microwave anechoic chamber, as shown in Fig. \ref{fig:6}. The tested antenna array is fixed to the center pillar of the rotating platform, and the receiving antenna has the same height as the transmitting antenna array. The distance between the antenna array and the receive antenna is 4.8 m, which satisfies the far-field test requirement.

In the experiments, we use a VNA to measure the  far-field amplitude and phase of the antenna under test (AUT), where the VNA model is R\&S ZNB. Specifically, the antenna array is connected to the output port of the beamforming control board, whose input port is linked to the power divider. The input of the power divider is connected to port one of the VNA, and the receiving antenna is connected to port two. The radiation field information of the AUT can thus be obtained by measuring the S21 parameter from the VNA.

In the practical testing environment, due to limitations of mechanical apparatus and testing equipment, we can only obtain the electric field data from the principal radiation plane (H-plane) of the antenna array. Therefore, we will use only these data to calculate the beamforming coefficients. Additionally, using electric field data from a single plane significantly reduces testing time and prevents severe oscillations of the experimental apparatus. To validate this approach, we compared the beamforming results using full-field data and using only single-plane electric field data in simulations, as shown in Fig. \ref{differentZ}. It shows that the directivity performance of the antenna array does not degrade significantly when only the H-plane electric field data is used. This observation can also be qualitatively explained from a mathematical perspective. When only the electric field data from the H-plane is used, the impedance coupling calculation formula can be written as:
\begin{align}
    z^{\rm Pla}_{mn} = \frac{1}{2\pi} \int_{-\pi}^\pi e^{jk(m-n)d\sin\phi} \, d\phi.
\end{align}
Revisiting the impedance coupling calculation method using full-field data \eqref{zmn}, we find that \eqref{zmn} essentially sums the plane impedance couplings across various $\theta$ planes, weighted by the factor $(|g(\theta, \phi)|^2 \sin\theta)$. For dipole antennas, this weighting factor becomes $(\sin\theta)^3$. Given this weighting relationship, it is evident that most of the impedance coupling information is captured when $\theta$ is 90°, allowing for the use of only H-plane electric field data to achieve satisfactory directivity performance. Limiting measurements to the H-plane also introduces certain limitations when analyzing and optimizing antenna performance:
\begin{itemize}
    \item By only measuring the electric field in the H-plane, we effectively simplify the analysis of the antenna's performance in the three-dimensional space. This simplification leads to the omission of electric field components present in the V-plane (vertical plane) and other non-horizontal planes. As a result, the complete spatial distribution of the electric field is not captured, which can lead to an incomplete understanding of the antenna's radiation characteristics. 
    \item The optimization process, when based solely on H-plane measurements, typically aims to maximize the directivity or gain in a specific direction within the H-plane. However, this approach might not ensure that the directivity in this direction is maximized when considering the full three-dimensional space. There is a risk that the antenna's performance might be optimized only for a specific plane, potentially overlooking the behavior of the antenna in other directions or planes. This could lead to suboptimal performance in terms of pattern integrity across the entire space.
\end{itemize}

We compare our proposed beamforming method with the MRT, the traditional superdirective beamforming method (field coupling ignored), and the simulation results. Fig. \ref{R2} shows the corresponding results when the number of antennas is four and the antenna spacing is $0.3\lambda$.  The MRT beamforming vector is obtained by measuring the channel coefficients using VNA. 
It can be found  that the MRT method has the widest 3-dB beamwidth, which equals ${\rm 82.52} ^{\circ}$, in the direction of interest. However, the narrowest 3-dB beamwidth, which equals 50.96$^{\circ}$, is achieved by our proposed beamforming method. From Fig. \ref{R2} (b), we may observe that although the traditional superdirective beamforming method also has a narrower 3-dB beamwidth compared to the MRT, its power level is higher than our method in almost all directions other than the main lobe. 
It can be found from Fig. \ref{R2} (c) that a good agreement is attained in both the main lobe and the side lobes.

Subsequently, the number of antennas was increased to eight, and the test results are shown in Fig. \ref{R53}. In Fig. \ref{R53} (a), it can be seen that with the increased number of antennas, both the proposed method and MRT show a reduction in 3dB beamwidth compared to using four antennas, which is consistent with theoretical expectations. The proposed method still achieves a narrower beamwidth. In Fig. \ref{R53} (b), the beamforming capability of traditional methods is significantly reduced due to the increased number of antennas. The sensitivity of the superdirective antenna array also increases, making it difficult for traditional methods that do not consider field coupling to achieve good directional performance with more antennas. In Fig. \ref{R53} (c), the proposed method still  matches the simulation results, although this consistency is weaker than with four antennas. This reduced consistency may be due to the increased number of antennas, which raises the higher precision requirements for the superdirective excitation coefficients. The quantization by the control board has also diminished some performance

By testing our fabricated prototype system of superdirective arrays, we have convincingly demonstrated the superior performance of our proposed method. Moreover, these tests confirm that our approach achieves superdirectivity in practical applications, aligning well with the simulation results.

\section{Conclusion}\label{V}
In this paper, we addressed the problem of beamforming coefficient calculation for superdirective antenna arrays. A double coupling-based scheme was proposed, where the antenna coupling effects are decomposed to impedance coupling and field coupling, which is obtained by a proposed full-wave simulation-based method. We proved the uniqueness of the obtained solution, which is able to fully characterize the distorted coupling field. 
Such a method is applicable in both electromagnetic simulation and practice (after the proposed adaptations).  Moreover, our methodology can be scaled up with more antennas and applied to other topologies, such as 16-elements array, 32-elements array, and planar arrays. 
We also developed a prototype of the superdirective dipole antenna array working at 1600 MHz and a hardware experimental platform to validate the proposed methods. The simulation results and experimental measurements both showed that our proposed method outperforms the state-of-the-art methods. 

\section*{Acknowledgement}The authors thank Thomas L. Marzetta for his helpful comments.

\appendices
\section{Proof of Theorem 1}\label{Appen. A}

The electric field radiated by the $z$-diretcted dipole antenna located at the origin can be expressed as\cite{balanis2015antenna}
\begin{align}\label{dipole_EF}
    \mathbf{e(r)} = j \eta \frac{k I_0 l e^{-j k r}}{4 \pi r} \sin \theta\hat{\theta},
\end{align}
where $\mathbf{r}=[r\sin\theta\cos\phi,r\sin\theta\sin\phi,r\cos\theta]$ is the far-field coordinate, $r=\|\mathbf{r}\|_2$ is the radius distance, $\hat{\theta}$ is the unit vector of $\theta$-direction. 
Since $r\gg d$, we have $r-(m-1)d\sin\theta\sin\phi\approx r$. Thus, Eq. \eqref{dipole_EF} can be rewritten as
\begin{align}\label{36}
    \mathbf{e}(\mathbf{r}_m)= 
    j \eta  \frac{kI_0le^{-j k (r-(m-1)d\sin\theta\sin\phi)}}{4\pi r} \sin \theta \hat{\theta}.
\end{align}
Letting $\mathbf{E}$ be the function group of the electric field radiated by antennas located at different positions, i.e., 
\begin{align}\label{37}
\mathbf{E}=\{\mathbf{e}(\mathbf{r}_1),\,\mathbf{e}(\mathbf{r}_2),\,\cdots,\,\mathbf{e}(\mathbf{r}_M)\}. \end{align}
We will prove that $\mathbf{E}$ is a group of linearly independent functions. It is notable that $\mathbf{e}(\mathbf{r}_m)$ is a function with two variables $\theta$ and $\phi$. The equation \eqref{37} can be rewritten as 
\begin{align}
    \mathbf{e}(\mathbf{r}_m)=Ae^{j k (m-1)d\sin\phi\sin\theta}\sin\theta\hat{\theta}
\end{align}
and
\begin{align}
    A = j \eta  \frac{kI_0l e^{-j k r}}{4\pi r}  .
\end{align}
Assuming that there exists a series of constants $k_i,i=1,\cdots,M$ such that 
\begin{align}\label{41}
    k_1\mathbf{e}(\mathbf{r}_1)+k_2\mathbf{e}(\mathbf{r}_2)+\cdots+k_M\mathbf{e}(\mathbf{r}_M)=0.
\end{align}
We want to show that the coefficients $k_i,i=1,\cdots,M$ are all zero. Differentiating the equation \eqref{41} with respect to $\sin\phi$, we obtain
\vspace{-0.1cm}
\begin{align}\label{42}
    k_2(jkd\sin\theta)\mathbf{e}(\mathbf{r}_2)+\cdots+k_M(jk(M-1)d\sin\theta)\mathbf{e}(\mathbf{r}_M)=0.
\end{align}
Continuing differentiating the above equation $(M-2)$ times with respect to $\sin\phi$ yields
\begin{align}
    k_M(jk(M-1)d\sin\theta)^{M-1}\mathbf{e}(\mathbf{r}_M) = 0,
\end{align}
which implicates that $k_M=0$. By backtracking the process we can obtain that $k_i=0,i=1,\cdots,M$. Thus, $\mathbf{E}$ is a group of linearly independent functions. Without loss of generality, take the calculation of the first column   of the coupling matrix as an example.  The coupling field is radiated by these antennas together, i.e. $\mathbf{e}_{{c}_1}\in {\rm span}\{\mathbf{e}_{{s}_n}:n\in[1,M]\} $. Since ${\rm rank}(\mathbf{E}_s)={\rm rank}(\mathbf{E}_s|\mathbf{e}_{{c}_1}) $, the linear function
\begin{align}
    \mathbf{e}_{{c}_1}=c_{11}\mathbf{e}_{{s}_1}+c_{21}\mathbf{e}_{{s}_2}+\cdots+c_{M1}\mathbf{e}_{{s}_M}
\end{align}
has a unique solution\cite{greub2012linear}. 

 Since $\mathbf{e}_{c_n}$ is uniquely determined by the field coupling matrix $\mathbf{C}$ and the electric field without distortion, the radiation pattern distortion can thus be fully characterized by $\mathbf{C}$. $\hfill\qedsymbol$

\ifCLASSOPTIONcaptionsoff
  \newpage
\fi


\bibliographystyle{IEEEtran}

    \bibliography{mybib}

%








\end{document}

%% file: defs.tex
\theoremstyle{plain}
\newtheorem{theorem}{Theorem}
\theoremstyle{plain}

\theoremstyle{plain}
\newtheorem{corollary}{Corollary}
\theoremstyle{plain}

\theoremstyle{plain}

\renewcommand{\baselinestretch}{1}